\pdfoutput=1
\RequirePackage[l2tabu,orthodox]{nag}
\documentclass[11pt,a4paper]{scrartcl}
\usepackage[margin=1in]{geometry}

\usepackage{lmodern}
\usepackage{flushend}
\usepackage[T1]{fontenc}
\usepackage[utf8]{inputenc}
\usepackage{microtype}
\usepackage{booktabs}
\usepackage{comment}

\usepackage{hyperref}

\renewcommand\path[1]{{\sffamily\footnotesize\detokenize{#1}}}
\usepackage{amsthm,amssymb,amsmath}

\usepackage{mathtools}
\DeclarePairedDelimiter\paren{\lparen}{\rparen}
\DeclarePairedDelimiter\abs{\lvert}{\rvert}
\DeclarePairedDelimiter\set{\{}{\}}
\DeclarePairedDelimiter\mset{\{\!\!\{}{\}\!\!\}}
\DeclarePairedDelimiterX\setc[2]{\{}{\}}{\,#1 \;\colon\; #2\,}
\DeclarePairedDelimiterX\parenc[2]{\lparen}{\rparen}{\,#1 \;\delimsize\vert\; #2\,}

\usepackage{mathrsfs}

\newcommand{\sgn}[1]{\operatorname{sgn}{#1}}

\newcommand{\extalg}[1]{\Lambda(#1)}
\newcommand{\CC}{\mathbf{C}}
\newcommand{\QQ}{\mathbf{Q}}
\newcommand{\NN}{\mathbf{N}}
\newcommand{\ZZ}{\mathbf{Z}}

\newcommand{\RR}{\mathbf{R}}

\newcommand{\bernoulli}[1]{\mathcal{B}^{k\times k}}
\newcommand{\var}[1]{\operatorname{Var}\left(#1\right)}

\usepackage{scalerel}

\newcommand{\poly}{\operatorname{poly}}
\newcommand{\Sub}{\operatorname{Sub}}
\let\Xi\varXi
\newcommand\algebra{{\mathcal{A}}}

\newcommand{\Hom}[2]{\operatorname{Hom}(#1\to#2)}
\newcommand{\InjHom}[2]{\operatorname{InjHom}(#1\to#2)}
\newcommand{\tw}[1]{\operatorname{tw}(#1)}
\newcommand{\pw}[1]{\operatorname{pw}(#1)}
\newcommand{\cone}{\mathfrak{c}}
\newcommand{\comp}{\mathfrak{a}}
\newcommand{\bag}{\mathfrak{b}}
\newcommand{\sep}{\mathfrak{s}}

\newcommand{\walkpol}{f}
\newcommand{\walks}{\mathscr W}
\newcommand{\paths}{\mathscr P}

\newcommand{\cano}[1]{\mathbf e_{#1}}
\newtheorem{thm}{Theorem}
\newtheorem{lem}[thm]{Lemma}
\newtheorem{prop}[thm]{Proposition}

\theoremstyle{definition}

\newtheorem*{rem}{Remark}

\newtheorem*{rep@theorem}{\rep@title}
\newcommand{\newreptheorem}[2]{%
	\newenvironment{rep#1}[1]{%
		\def\rep@title{#2 \ref{##1}}%
		\begin{rep@theorem}[restated]}%
		{\end{rep@theorem}}}

\newreptheorem{thm}{Theorem}
\newreptheorem{lem}{lem}
\newreptheorem{prop}{Proposition}
\newreptheorem{cor}{Corollary}

\usepackage{enumitem}
\newenvironment{algor}[3]{%
\bigskip
\noindent{\sffamily\bfseries Algorithm #1} ({\itshape#2\/}) {\itshape #3}
\begin{description}[noitemsep,labelindent=0pt,labelwidth=1.7em,labelsep=0pt,leftmargin=!]%
\vspace{-.8ex}}{%
\end{description}\medskip}

\usepackage{authblk}

\author[a]{Cornelius Brand}
\author[a]{Holger Dell}
\author[b]{Thore Husfeldt}

\affil[a]{%
  Saarland University and Cluster of Excellence (MMCI), Saarbr\"{u}cken, Germany%
  \\{\{cbrand,hdell\}@mmci.uni-saarland.de}
}
\affil[b]{%
  Lund University and
  Basic Algorithms Research Copenhagen, ITU Copenhagen%
  \\{thore@itu.dk}
}

\title{Extensor-Coding}

\begin{document}

\maketitle

\begin{abstract}
  We devise an algorithm that approximately computes the number of paths of length~$k$ in a given directed graph with~$n$ vertices up to a multiplicative error of~$1 \pm \varepsilon$.
  Our algorithm runs in time $\varepsilon^{-2} 4^k(n+m) \poly(k)$.
  The algorithm is based on associating with each vertex an element in the exterior (or, Grassmann) algebra, called an extensor, and then performing computations in this algebra.
  This connection to exterior algebra generalizes a number of previous approaches for the longest path problem and is of independent conceptual interest.
  Using this approach, we also obtain a deterministic $2^{k}\cdot\poly(n)$ time algorithm to find a~$k$-path in a given directed graph that is promised to have few of them.
  Our results and techniques generalize to the subgraph isomorphism problem when the subgraphs we are looking for have bounded pathwidth.
  Finally, we also obtain a randomized algorithm to detect $k$-multilinear terms in a multivariate polynomial given as a general algebraic circuit.
  To the best of our knowledge, this was previously only known for algebraic circuits not involving negative constants.
\end{abstract}

\section{Introduction}

A path is just a walk that does not vanish in the exterior algebra.
This observation leads us to a new approach for algebraic graph algorithms for the $k$-path problem, one of the benchmarks of progress in parameterized algorithms.
Our approach generalizes and unifies previous techniques in a clean fashion, including the color-coding method of Alon, Yuster, and Zwick~\cite{DBLP:journals/jacm/AlonYZ95} and the vector-coding idea of Koutis~\cite{DBLP:conf/icalp/Koutis08}.
Color-coding yields a randomized algorithm for approximately counting $k$-paths~\cite{alon2008biomolecular} that runs in time $(2e)^k\poly(n)$.
We improve the running time to $4^k\poly(n)$, addressing an open problem in the survey article of Koutis and Williams~\cite{Koutis:2015:AFF:2859829.2742544}.
Our approach applies not only to paths, but also to other subgraphs of bounded pathwidth.

In hindsight, it is obvious that the exterior algebra enjoys exactly the properties needed for the $k$-path problem.
Thus, it seems strange that this construction has eluded algorithms designers for so long.
But as the eminent combinatorialist Gian-Carlo Rota observed in 1997, ``[t]he neglect of the exterior algebra is the mathematical tragedy of our century,''~\cite{Rota} so we are in good company.

The exterior algebra is also called alternating algebra, extended algebra, or Grassmann algebra after its 19th century discoverer.
It is treated extensively in any modern textbook on algebra, and has applications in many fields, from differential geometry and representation theory to theoretical physics.
Conceptually, our contribution is to identify yet another entry in the growing list of applications of the exterior algebra, inviting the subgraph isomorphism problem to proudly take its place between simplicial complexes and supernumbers.

\paragraph{Longest Path.}
The Longest Path problem is the optimization problem to find a longest (simple) path in a given graph.
Clearly, this problem generalizes the NP-hard Hamiltonian path problem~\cite{Garey:1979:CIG:578533}.
We consider the decision version, the $k$-path problem, in which we wish to find a path of length~$k$ in a given graph~$G$.
It was proved fixed-parameter tractable \emph{avant la lettre}~\cite{MONIEN1985239}, and a sequence of both iterative improvements and conceptual breakthroughs~\cite{DBLP:journals/jal/Bodlaender93,DBLP:journals/jacm/AlonYZ95,DBLP:journals/siamcomp/Bjorklund14,DBLP:conf/wg/KneisMRR06,doi:10.1137/080716475,DBLP:journals/jacm/FominLPS16,DBLP:journals/ipl/Williams09} have lead to the current state-of-the-art for undirected graphs: a randomized algorithm by Björklund \emph{et al.}~\cite{DBLP:journals/jcss/BjorklundHKK17} in time~$1.66^k\cdot \poly(n)$.
For directed graphs, the fastest known randomized algorithm is by Koutis and Williams~\cite{koutis2009limits} in time $2^k\cdot\poly(n)$, whereas the fastest deterministic algorithm is due to Zehavi~\cite{Zehavi14} in time~$2.5961^k\cdot\poly(n)$.

\paragraph{Subgraph isomorphism.}
The subgraph isomorphism problem generalizes the $k$-path problem and is one of the most fundamental graph problems~\cite{cook1971complexity,ullmann1976algorithm}:
Given two graphs~$H$ and~$G$, decide whether~$G$ contains a subgraph isomorphic to~$H$.
This problem and its variants have a vast number of applications, covering areas such as statistical physics, probabilistic inference, and network analysis~\cite{Milo824}.
For example, such problems arise in the context of discovering \emph{network motifs}, small patterns that occur more often in a network than would be expected if it was random.
Thus, one is implicitly interested in the counting version of the subgraph isomorphism problem:
to compute the \emph{number} of subgraphs of~$G$ that are isomorphic to~$H$.
Through network motifs, the problem of counting subgraphs has found applications in the study of gene transcription networks, neural networks, and social networks~\cite{Milo824}. Consequently, there is a large body of work dedicated to algorithmic discovery of network motifs~\cite{grochow2007network,alon2008biomolecular,omidi2009moda,kashani2009kavosh,schreiber2005frequency,chen2006nemofinder,kashtan2004efficient,wernicke2006efficient,DBLP:conf/alcob/SchillerJHS15}.
For example, Kibriya and Ramon~\cite{DBLP:journals/datamine/KibriyaR13,DBLP:books/crc/p/RamonCCW14} use the ideas of Koutis and Williams~\cite{koutis2009limits} to enumerate all trees that occur frequently.

\paragraph{Counting subgraphs exactly.}
The complexity of exact counting is often easier to understand than the corresponding decision or approximate counting problems.
For instance, the counting version of the famous dichotomy conjecture by Feder and Vardi~\cite{DBLP:conf/stoc/FederV93,DBLP:journals/siamcomp/FederV98} was resolved by Bulatov~\cite{DBLP:conf/icalp/Bulatov08,DBLP:journals/jacm/Bulatov13} almost a decade before proofs were announced for the decision version by Bulatov~\cite{DBLP:journals/corr/Bulatov17a} and Zhuk~\cite{DBLP:journals/corr/Zhuk17}.
A similar phenomenon can be observed for the parameterized complexity of the subgraph isomorphism problem, the counting version of which is much better understood than the decision or approximate counting versions:
The problem of counting subgraphs isomorphic to~$H$ is fixed-parameter tractable if~$H$ has a vertex cover of bounded size~\cite{williams2013finding} (also cf.\ \cite{DBLP:journals/siamdm/KowalukLL13,curticapean2014complexity,DBLP:conf/stoc/CurticapeanDM17}), and it is \#W[1]-hard whenever~$H$ is from a class of graphs with unbounded vertex cover number~\cite{curticapean2014complexity,DBLP:conf/stoc/CurticapeanDM17}, and thus it is not believed to be fixed-parameter tractable in the latter case.
In particular, this is the case for counting all $k$-paths in a graph.
The fastest known general-purpose algorithm~\cite{DBLP:conf/stoc/CurticapeanDM17} for counting $H$-subgraphs in an $n$-vertex graph~$G$ runs in time~$k^{O(k)} n^{t^\ast+1}$ where $k$ is the number of vertices of~$H$ and $t^\ast$ is the largest treewidth among all homomorphic images of~$H$.

\paragraph{Our results.}
For finite directed or undirected graphs~$H$ and~$G$, let $\Sub(H,G)\in\NN$ be the number of (not necessarily induced) subgraphs of~$G$ that are isomorphic to~$H$.
The main algorithmic result in this paper is a randomized algorithm that computes an approximation to this number.
\def\statethmapproxcount{
  There is a randomized algorithm that is given two graphs~$H$ and~$G$, and a number~${\varepsilon>0}$ to compute an integer~$\tilde N$ such that, with probability~$99\%$, 
\begin{equation}
  (1-\varepsilon) \cdot\Sub(H,G) \le
  \tilde N
  \le (1+\varepsilon) \cdot\Sub(H,G)
  \,.
\end{equation}
This algorithm runs in time $\varepsilon^{-2}\cdot 4^k n^{\pw H+1}\cdot \poly(k)$,
where~$H$ has~$k$ vertices and pathwidth~$\pw H$, and~$G$ has~$n$ vertices.
}
\begin{thm}[Approximate subgraph counting]
  \label{thm approx count sub}
  \statethmapproxcount
\end{thm}
Our algorithm works for directed and undirected graphs with the same running time (in fact, undirected graphs are treated as being bi-directed).
An algorithm such as the one in Theorem~\ref{thm approx count sub} is called a fixed-parameter tractable randomized approximation scheme (FPT-RAS) for $\Sub$.
The notion of an FPT-RAS was defined by Arvind and Raman~\cite{DBLP:conf/isaac/ArvindR02}, who use a sampling method based on Karp and Luby~\cite{DBLP:conf/focs/KarpL83} to obtain a version of Theorem~\ref{thm approx count sub} with an algorithm that runs in time $\exp(O(k\log k))\cdot n^{\tw H+O(1)}$.
For the special cases of paths and cycles, Alon and Gutner~\cite{DBLP:conf/iwpec/AlonG09,DBLP:journals/talg/AlonG10} are able to combine the color-coding technique by Alon, Yuster, and Zwick~\cite{DBLP:journals/jacm/AlonYZ95} with balanced families of hash functions to obtain an algorithm for approximately counting paths or cycles in time~$\exp(O(k\log\log k))\cdot n\log n$.
Alon \emph{et al.}~\cite{alon2008biomolecular}, in turn, use the color-coding technique to obtain the first singly-exponential time version of Theorem~\ref{thm approx count sub}, in particular with an algorithm running in time~$\varepsilon^{-2} \cdot (2e)^k\cdot n^{\tw H+O(1)}$.
To the best of our knowledge, Theorem~\ref{thm approx count sub} is now the fastest known algorithm to approximately count subgraphs of small pathwidth.

When we are promised that $G$ contains not too many subgraphs isomorphic to~$H$, we obtain the following deterministic algorithm.

\begin{thm}[Detecting subgraphs when there are few]
  \label{thm decide sub}
  There is a deterministic algorithm that is given two graphs~$H$ and~$G$ to decide whether~$G$ has a subgraph isomorphic to~$H$, with the promise that~$G$ has at most $C\in\NN$ such subgraphs.
  This algorithm runs in time $O(C^2 2^k n^{\pw H+O(1)})$,
  where the number of vertices of~$H$ is~$k$ and the number of vertices of~$G$ is~$n$.
\end{thm}

Without the promise on the number of subgraphs,
Fomin \emph{et al.}~\cite{DBLP:journals/jcss/FominLRSR12} detect subgraphs in randomized time $\tilde O(2^kn^{\tw H+1})$ and
Fomin \emph{et al.}~\cite{DBLP:journals/jacm/FominLPS16} do so in deterministic time $2.619^k n^{O(\tw H)}$.
For $C\le O(1)$, or $C\le \poly(n,k)$ when ignoring polynomial factors, we thus match the running time of the fastest randomized algorithm, but do so deterministically, and for $C\le O\paren{1.144^k}$, our algorithm is the fastest deterministic algorithm for this problem.
For the interesting special case of paths, the running time of the fastest deterministic algorithm for undirected or directed $k$-paths (without promise) is $2.5961^k \cdot \poly(n)$ by Zehavi~\cite{Zehavi14}, which we improve upon if $C\le O(1.139^k)$.

Our method also applies to the problem of detecting whether a multivariate polynomial contains a multilinear term.
\begin{thm}[Detecting multilinear terms]
  \label{thm detect multilinear}
  Given an algebraic circuit~$C$ over $\ZZ[\zeta_1,\dots,\zeta_n]$ and a number~$k$, we can detect whether the polynomial $C(\zeta_1,\dots,\zeta_n)$ has a degree-$k$ multilinear term in randomized time~$4.32^k\cdot\abs{C}\cdot \poly(n)$.
\end{thm}

Using algebraic fingerprinting with elements from a group algebra, Koutis and Williams~\cite{DBLP:conf/icalp/Koutis08,koutis2009limits} can do this in randomized~$2^k \cdot \poly(n)$ time for monotone algebraic circuits, that is, circuits that do not involve negative values.
Working over an algebra whose ground field of characteristic~$0$, we are able to remove the requirement that the circuit is free of cancellations in Theorem~\ref{thm detect multilinear}.
To the best of our knowledge, this is the first fixed-parameter tractable algorithm for the problem of detecting a $k$-multilinear term in the polynomial computed by a general algebraic circuit.
Our algorithm uses color-coding and performs the computation in the exterior algebra over~$\QQ^k$.
To reduce the running time from $2^ke^k \cdot \poly(n)$ to $4.32^k \cdot \poly(n)$, we use an idea of Hüffner, Wernicke, and Zichner~\cite{DBLP:journals/algorithmica/HuffnerWZ08}, who improved color-coding by using $1.3\cdot k$ instead of only~$k$ different colors.

\paragraph{Related hardness results.}
Under the exponential-time hypothesis (ETH) by Impagliazzo and Paturi~\cite{IP01}, the running time of the algorithm in Theorem~\ref{thm approx count sub} is optimal in the following asymptotic sense:
The exponent of~$n$ cannot be improved since $f(k) n^{o(t)}$ time is impossible
even in the case that~$H$ is a~$k$-clique~\cite{chen2004tight}, where $t=k-1$.
Likewise, a running time of the form~$\exp(o(k))\cdot\poly(n)$ is impossible even in the case that~$t=1$, since this would imply an $\exp(o(n))$ time algorithm for the Hamiltonian cycle problem and thereby contradict~ETH~\cite{ImpagliazzoPZ01}.
Moreover, the factor $\varepsilon^{-2}$ in the running time stems from an application of Chebyshev's inequality and is unlikely to be avoidable.

\subsection{Organization}
In the body text of the present manuscript, we focus entirely on paths instead of general subgraphs~$H$.
Section~\ref{sec: Lambda} contains an elementary development of the exterior algebra, deliberately eschewing abstract algebra.
Section~\ref{sec: Extensor-coding} then presents a number of different extensor-codings and establishes Theorems~\ref{thm approx count sub} and \ref{thm decide sub} for the case where the pattern graph $H$ is a $k$-path:
Theorem~\ref{thm approx count sub} corresponds to Algorithm~C and Theorem~\ref{thm: algorithm C} in Section~\ref{sec: Bernoulli coding};
Theorem~\ref{thm decide sub} corresponds to Algorithm~F and Theorem~\ref{thm: algorithm F} in Section~\ref{sec: Edge-Variables}.
Section~\ref{sec: Previous} is mainly expository and connects our approach to previous work.
The technical details needed to establish Theorems~\ref{thm approx count sub}--\ref{thm detect multilinear} in full generality are moved to the appendices.

\subsection{Graphs and Walks}

Let $G$ be a directed graph with $n$ vertices and $m$ edges.
The set of vertices is $V(G)$ and enumerated as  $\{v_1,\ldots, v_n\}$.
The set of edges is $E(G)$, the edge from $u$ to $v$ is denoted by~$uv$.
A sequence of vertices $w_1,\ldots,w_k$ in $V(G)$ such that $w_iw_{i+1} \in E$~holds for all $i \in \{1,\dots,k-1\}$
is called a $k$-\emph{walk} in $G$.
A walk of distinct vertices is called a \emph{path}. 
The set of $k$-walks is denoted by~$\walks$ and the set of $k$-paths is denoted by~$\paths$.
We write $\poly(n)$ for the set of polynomially bounded functions in~$n$.
Throughout the document, we silently assume $k \leq n$.

Let $R$ be a ring and consider a mapping $\xi\colon V(G)\cup E(G)\to R$.
The \emph{walk-sum} $\walkpol(G;\xi)$ of~$\xi$ is defined via
\begin{align}
\label{eq: walk-sum def}
  \walkpol(G;\xi) = \sum_{w_1\dots w_k\in\walks}
  \xi(w_1) \xi(w_1w_2) \xi(w_2)
  \cdots
  \xi(w_{k-1}) \xi(w_{k-1}w_k) \xi(w_{k})
  \,,
\end{align}
evaluated in $R$.
As a matter of folklore, the walk-sum can be evaluated with $O(kn^2)$ operations over $R$ using using a well-known connection with powers of the adjacency matrix:
\begin{equation}\label{eq: walk-sum matrixpower}
  \walkpol(G;\xi)
  =
    \begin{pmatrix}
      1 \hdots 1
    \end{pmatrix}
    \cdot
    A^{k-1}
    \cdot
    \begin{pmatrix}
      \xi(v_1)\\\vdots\\\xi(v_n)
    \end{pmatrix}
  \,,
\end{equation}
where $A$ is the $n\times n$ matrix whose $vw$-entry is given by
\begin{equation}\label{eq:abstract adjacency matrix}
  a_{vw}
  =
  \begin{cases}
    \xi(v)\xi(vw), &\text{if $vw\in E(G)$;}\\
    0,       &\text{otherwise.}
  \end{cases}
\end{equation}
Note that the expression for $f(G;\xi)$ in \eqref{eq: walk-sum matrixpower} can be evaluated in such a way that every product in $R$ has the form $x\cdot y$ where $y$ belongs to the range of~$\xi$ (rather than all of $R$).
Moreover, we assume input graphs to be given as adjacency lists, in which case the expression in~\eqref{eq: walk-sum matrixpower} can be evaluated with $O(k (n+m))$ operations over~$R$, since the product of an $m$-sparse matrix and a vector can be computed with $O(n+m)$ operations over~$R$ (equivalently, we can view this process as a distributed algorithm that computes~$(A^{k-1}\cdot(\xi(v_1)\dots\xi(v_n))^T)_v$ at each vertex~$v$ in $k-1$ rounds of synchronized communication).
If $\xi\colon V(G)\to R$ is a partial assignment, we silently extend it to a full assignment by setting the remaining variables to~$1\in R$.

\section{The Exterior Algebra}
\label{sec: Lambda}
\subsection{Concrete Definition}
\label{sec: concrete introduction}

We now give an elementary and very concrete definition of the exterior algebra, and recall the properties of the wedge product.
Readers familiar with this material can skip Section~\ref{sec: concrete introduction}. 

Let $F$ be a field, $k$ be a positive integer, and let~$\cano1,\dots,\cano k$ be the canonical basis of the $k$-dimensional vector space~$F^k$.
Every element $a$ of $F^k$ is a linear combination $a_1\cano 1+\cdots+a_k\cano k$ with field elements $a_1,\ldots,a_k\in F$.
We sometimes write $a$ as the column vector $(a_1,\ldots,a_k)^T$.
Addition and scalar multiplication are defined in the usual way.

We extend $F^k$ to a much larger, $2^k$-dimensional vector space $\Lambda(F^k)$ as follows.
Each basis vector $\cano I$ of $\Lambda(F^k)$ is defined by a subset $I$ of indices from $\{1,\ldots, k\}$.
The elements of $\Lambda(F^k)$ are called \emph{extensors}.
Each element is a linear combination $\sum_{I\subseteq \{1,\ldots,k\}} a_I \cano I$ of basis vectors.
We turn $\Lambda(F^k)$ into a  vector space by defining addition and scalar multiplication in the natural fashion.
For instance, if $F$ is the rationals, typical elements in $\Lambda(F^k)$ with $k=3$ are
\(x= 3 \cano {\{1, 2\}} - 7 \cano {\{3\}}\) and $y = \cano {\{1\}} + 2 \cano {\{3\}}$  and we have
$x+2y=  3 \cano {\{1,3\}}+ 2 \cano {\{1\}} - 3 \cano {\{3\}}$.
By confusing $\cano i$ with $\cano {\{i\}}$ for $i\in \{1,\dots,k\}$, we can view $F^k$ as a subspace of $\Lambda(F^k)$ spanned by the singleton basis vectors.
This subspace is sometimes called $\Lambda^1(F^k)$, the set of \emph{vectors}.
The element $\cano \emptyset$ is just $1$ in the underlying field, so $\Lambda^0(F^k)= F$.
In general, $\Lambda^i(F^k)$ is the set of extensors spanned by basis vectors $\cano I$ with $|I|=i$, sometimes called \emph{$i$-vectors}.
Of particular interest is $\Lambda^2(F^k)$, the set of \emph{blades} (also called bivectors).

To turn $\Lambda(F^k)$ into an \emph{algebra}, we define a multiplication $\land$ on the elements of $\Lambda(F^k)$.
The multiplication operator we define is called the \emph{wedge product} (also called exterior or outer product)
and the resulting algebra is called the \emph{exterior algebra}.
We require $\land$ to be associative
\[ (x\land y)\land z= x\land(y\land z)\]
and bilinear
\begin{align*}
x \land (a\cdot y+z) &= a\cdot x\land y + x\land z\,,\\
(x+a\cdot y) \land z &= x \land z + a\cdot y \land z\,,
\end{align*}
for all  $a\in F$ and  $x,y,z\in \Lambda(F^k)$.
Thus, it suffices to define how $\land$ behaves on a pair of basis vectors $\cano I$ and $\cano J$.
If $I$ and $J$ contain a common element, then we set $\cano I\land \cano J=0$.
Otherwise, we set $\cano I\land \cano J= \pm \cano {I\cup J}$; it only remains to define the sign, which requires some delicacy.
(The intuition is that we want $\land$ to be anti-commutative on $F^k$, that is, $x\land y=-y\land x$ for $x,y\in F^k$.)
Write $I=\{i_1,\ldots, i_r\}$ and $J=\{j_1,\ldots, j_s\}$, both indexed in increasing order.
Then we define
\[ \cano I\land \cano J = (-1)^{\sgn(I, J)} \cano {I\cup J}\,,\]
where $\sgn (I,J)$ is the sign of the permutation that brings the sequence
\(i_1,\ldots,i_r,j_1,\ldots, j_s\)
into increasing order.

For instance, if $\max I <\min J$, then there is nothing to permute, so $\cano 1\land \cano 2 = \cano {\{1,2\}}$. 
Consequently, we now abandon the set-indexed notation $\cano {\{i_1,\ldots, i_r\}}$ (where $i_1<\cdots < i_r$) and just write $\cano {i_1}\land\cdots\land \cano {i_r}$ instead.
It is also immediate that $\cano 1\land \cano 2 = -\cano 2\land \cano 1$.
In general, we can multiply basis vectors using pairwise transpositions and associativity, \emph{e.g.},
\( (\cano 1\land \cano 3\land \cano 6)\land (\cano 2\land \cano 4) =
 -\cano 1\land \cano 3\land \cano 2\land \cano 6\land \cano 4 =
 \cano 1\land \cano 2\land \cano 3\land \cano 6\land \cano 4 =
 -\cano 1\land \cano 2\land \cano 3\land \cano 4\land \cano 6\,.\)

\subsection{Properties}

 The wedge product on $F^k$ has the following properties:
 \begin{enumerate}[label=(W\arabic*)]
\item\label{vector alternating} \emph{Alternating on vectors.} By its definition, the wedge product enjoys anticommutativity on the basis vectors of $F^k$, which is to say $\cano i \land \cano j = -\cano j \land \cano i$.
Employing bilinearity, this directly translates to any two vectors $x,y \in F^k$, meaning $x \land y = -y \land x$ holds,
whereby $x \land x$ vanishes.  
\item\label{decomposable alternating} \emph{Alternating on decomposable extensors.}
    An extensor $x\in\Lambda(F^k)$ is \emph{decomposable} if there are vectors $v_1,\ldots,v_r\in F^k$ satisfying $x=v_1\land \cdots\land v_r$.
     Every extensor in $\Lambda^i(F^k)$ is decomposable for $i\in\{0,1,k-1,k\}$, but not all extensors are decomposable: $\cano 1\wedge \cano 2+\cano 2\wedge \cano 4\in \Lambda^2(F^4)$ is an example.
      The previous property extends to decomposable vectors: If the extensors $x_1,\cdots, x_r$ are decomposable and two of them are equal, then it follows from Property~\ref{vector alternating} that their wedge product $x_1\land\cdots\land x_r$ vanishes.
\item\label{wedge determinant}  \emph{Determinant on $F^{k\times k}$.}
  For $k=2$ write $x,y\in F^2$ as column vectors $(x_1,x_2)$ and $(y_1,y_2)$.
  Elementary calculations show
\( x\land y = (x_1 y_2 - y_1 x_2) \cdot \cano 1 \land \cano 2\,,\)
and we recognize the determinant of the $2\times 2$-matrix whose columns are~$x$ and~$y$.
This is not a coincidence. Since $\Lambda^k(F^k)$ is linearly isomorphic to $F$---indeed, $\Lambda^k(F^k) = F \cdot (\cano 1 \land \cdots \land \cano k)$---we can understand the map taking $(x_1,\ldots,x_k)$ to $x_1 \land \cdots \land x_k \in \Lambda^k(F^k) \cong F$ as a multilinear form, which by virtue of the previous properties is alternating and sends $(\cano 1,\ldots,\cano k)$ to $1$. These properties already characterize the determinant among the multilinear forms.
With this, we have arrived at a fundamental property of the exterior algebra.
Let $x_1,\dots,x_k\in F^k$ and write
\[
x_1= \begin{pmatrix} x_{11} \\ \vdots \\ x_{k1} \end{pmatrix}
\,,\qquad
\ldots
\,,
\qquad
x_k= \begin{pmatrix} x_{1k} \\ \vdots \\ x_{kk} \end{pmatrix}\,.
\]
The wedge product of~$x_1,\dots,x_k$ exhibits a determinant:
\begin{equation}\label{eq: wedge as det}
  x_1\land\cdots\land x_k = \det
  \begin{pmatrix}
    x_{11} & \cdots & x_{1k} \\
    \vdots & \ddots & \vdots\\ 		
    x_{k1} & \cdots & x_{kk} \\
  \end{pmatrix}
  \cdot \cano {[k]}
  \,,
\end{equation}
where we use the shorthand $\cano {[k]}$ for the highest-grade basis extensor $\cano 1 \land \cdots\land  \cano k$.
\end{enumerate}

To avoid a misunderstanding: Neither of these properties extends to all of $\Lambda(F^k)$.
For instance, if $x = \cano 1 \land \cano 3 + \cano 2$ then
$ x\land x =
 (\cano 1\land \cano 3 + \cano 2)\land(\cano 1\land \cano 3 + \cano 2) =
  \cano 1\land \cano 3\land \cano 1\land \cano 3 + \cano 1\land \cano 3\land \cano 2 + \cano 2\land \cano 1\land \cano 3 + \cano 2\land \cano 2 =
  0 - \cano 1\land \cano 2\land \cano 3 - \cano 1\land \cano 2\land \cano 3 + 0 =
  -2\cdot \cano 1\land \cano 2\land \cano 3\neq 0$.

\subsection{Representation and Computation}
\label{sec: rep and comp}

We represent an extensor $x\in\Lambda(F^k)$ by its coefficients in the expansion $x=\sum_{I\subseteq \{1,\ldots, k\}} x_I\cano I$, using $2^k$ elements~$x_I$ from $F$.
The sum $z=x+y$ is given by coefficient-wise addition $z_I = x_I+y_I$, requiring~$2^k$ additions in $F$.
The wedge product $z=x\wedge y$ is
\[
  \left(\sum_{I\subseteq K}x_I \cano I \right)\wedge\left(\sum_{J\subseteq K} y_J \cano J\right) =
  \sum_{I,J\subseteq K}x_Iy_J\cdot \cano I\wedge \cano J \,.
\]
When $y$ belongs to $\Lambda^j(F^k)$, we can restrict the summation to subsets~$J$ with $|J|=j$.
Thus, $x\wedge y$ for $x\in\Lambda(F^k)$ and $y\in\Lambda^j(F^k)$ can be computed using $2^k\binom{k}{j}$ multiplications in $F$.
This is the only wedge product we need for our results, and only for $j\in\{1,2\}$.

In particular, $\Lambda(F^k)$ is a ring with multiplication~$\wedge$.
Then, for a mapping $\xi\colon V(G)\rightarrow \Lambda^j(F^k)$, we can compute the walk-sum $f(G;\xi)$ from \eqref{eq: walk-sum def} using $O(n+m) 2^k\binom{k}{j}$ field operations, which is $(n+m)2^k\poly(k)$ for $j=O(1)$.

For completeness, the case where $y\in\Lambda(F^k)$ is a general extensor, can be computed faster than $4^k$.
By realizing that the coefficient $z_I$ is given by the \emph{alternating subset convolution}
\begin{equation}\label{eq: alternating subset convolution}
  z_I = \sum_{J\subseteq I} (-1)^{\sgn(J,I\setminus J)} x_J y_{I\setminus J}\,,
\end{equation}
we see that $x\wedge y$ can be computed in $3^k$ field operations. 
By following Leopardi \cite{leopardi2005generalized} and the subsequent analysis of W\l{}odarczyk \cite{DBLP:conf/iwpec/Wlodarczyk16}, this bound can be improved to $O^\ast(2^{\omega\frac{k}{2}})$, where $\omega$ is the exponent for matrix multiplication. 
This works by making use of an efficient embedding of a Clifford algebra related to $\Lambda(F^{k})$ into a matrix algebra of dimension ${2^{k/2} \times 2^{k/2}}$, and expressing one product in $\Lambda(F^{k})$ as $k^2$ products in this Clifford algebra. 
(We never need this.)

\section{Extensor-coding}
\label{sec: Extensor-coding}

\begin{table*}
  \renewcommand{\arraystretch}{1.1}
  \resizebox{\columnwidth}{!}{%
\begin{tabular}{llllll}
  \multicolumn{2}{l}{ Name}             &  $v_i\mapsto$                                       & $e\mapsto$ & Algebra & Section \\\midrule
  $\phi$           & Vandermonde        &  $(i^0,\ldots,i^{k-1})^T$                           & $1$      & $\Lambda(F^k)$ & \ref{sec: Vandermonde Vectors}, \ref{sec: baseline} \\
  $\overline\phi$  & Lifted Vandermonde &  $\overline{\phi(v_i)}$                             & $1$      & $\Lambda(F^{2k})$  & \ref{sec: deterministic alg}\\
  $\overline\beta$ & Lifted Bernoulli   &  $\overline{(\pm 1, \ldots, \pm 1)^T}$              & $1$      & $\Lambda(F^{2k})$ & \ref{sec: Bernoulli coding}\\
  $\eta$           & Edge-variable      &  $\phi(v_i)$                                        & $y_e$    & $\Lambda(F^k)[Y]$ & \ref{sec: Edge-Variables} \\
  $\rho$           & Random edge-weight &  $\phi(v_i)$                                        & Random $r\in\{1,\ldots,100k\}$ & $\Lambda(F^k)$ & \ref{sec: Random Edge-Weights} \\
  $\lambda$        & Labeled walks        &  $(x_i^{(1)},\ldots,x_i^{(k)})^T$                           & $y_e$      & $\Lambda(F^k)[X,Y]$ & \ref{sec: Labeled Walks} \\
  $\overline\chi$       & Color-coding       &  $\overline{\cano j}$, random $j\in\{1,\ldots,k\}$  & $1$      & $Z(F^k) \subset \Lambda(F^{2k})$  & \ref{sec: color-coding}\\
\end{tabular}
}
\caption{\small Extensor-codings of graphs used in this paper.}
\end{table*}

\subsection{Walk Extensors}
\label{sec: Walk Extensors}
An \emph{extensor-coding} is a mapping $\xi\colon V(G)\rightarrow \Lambda(F^k)$ associating an extensor with every vertex of~$G$.
If~$W$ is a walk~$w_1\dots w_\ell$ of length~$\ell$ in~$G$, then we define the \emph{walk extensor} $\xi(W)$ as
\[ \xi(W) = \xi(w_1)\land\cdots\land\xi(w_\ell)\,. \]
Suppose now that $\xi$ always maps to decomposable extensors.
We can formulate our main insight:

\begin{lem}\label{lem: main}%
  If $\xi(v)$ is decomposable for all $v\in V(G)$  and $W$ is not a path, then $\xi(W)=0$.
\end{lem}
\begin{proof}
Directly follows from Property~\ref{decomposable alternating}.
\end{proof}

In particular, the (easily computed) walk-sum of $\xi$ over the ring~$R$ with~$R=\Lambda(F^k)$ is a sum over paths:
\begin{equation}\label{eq:walkpol = pathpol}
  \walkpol(G;\xi) =
  \sum_{W\in\walks} \xi(W) =
    \sum_{P\in\paths} \xi(P)\,.
  \end{equation}

  We can view $\xi$ as the $(k\times n)$ matrix~$\Xi$ over~$F$ consisting of the columns $\xi(v_1),\dots,\xi(v_n)$.
  By~\eqref{eq: wedge as det}, we have 
  \begin{equation}\label{eq: d for xi}
    \xi(w_1\dots w_k) = d\cdot \cano {[k]}\,,
  \end{equation}
  where~$d$ is the determinant of the $(k\times k)$-matrix $\Xi_P$ of columns $\xi(w_i),\dots, \xi(w_k)$.
  This matrix is a square submatrix of~$\Xi$, and vanishes if two columns are the same.

  \medskip
  While it is terrific that non-paths vanish, we are faced with the dangerous possibility that $\walkpol(G;\xi)$ vanishes as a whole, even though~$\paths$ is not empty.
  There are two 
  distinct reasons why this might happen:
  the extensor~$\xi(P)$ might vanish for a path~${P\in\paths}$, or the sum of non-vanishing extensors~$\xi(P)$ vanishes due to cancellations in the linear combination.

  \subsection{Vandermonde Vectors}
  \label{sec: Vandermonde Vectors}

  To address the first concern, we consider an extensor-coding $\xi$ in \emph{general position}, that is, such that $\xi(w_1\dots w_k)\ne 0$ for all \mbox{$k$-tuples} of distinct vertices~$w_1\dots w_k$.
  Thus, $\xi$ is in general position if and only if all square submatrices of~$\Xi$ are non-singular.
  Rectangular Vandermonde matrices have this property.

  \begin{lem}\label{lem: vandermonde extensor}
    Let the \emph{Vandermonde extensor-coding~$\phi$ of~$G$} be
    \begin{equation}\label{eq:Vandermonde extensor}
      \phi(v_i)
      =
      (1,i^1,i^2,\dots,i^{k-1})^T\text{ for all $i\in\set{1,\dots,n}$}\,.
    \end{equation}
    If~$i_1,\dots,i_k\in\set{1,\dots,n}$, then
    \[ \phi(v_{i_1}\dots v_{i_k})=\det \Phi_P \cdot \cano {[k]}\,,\]
    where
      \begin{equation}\label{eq: Vandermonde matrix}
      \Phi_P=\begin{pmatrix}
        1 & 1 & \hdots & 1 \\
        i_1 & i_2 & \hdots & i_k \\
        \vdots&\vdots &\ddots&\vdots   \\
        i_1^{k-1} & i_2^{k-1}&\hdots & i_k^{k-1}
      \end{pmatrix}\,.
      \end{equation}
      In particular,
    \begin{equation}\label{eq:Vandermonde determinant}
      d = \det \Phi_P
      =
      \prod_{\substack{i_a,i_b\\a<b}} (i_a-i_b)
      \,.
    \end{equation}
  \end{lem}

\subsection{Baseline Algorithm}
\label{sec: baseline}

  Our second concern was that distinct non-vanishing paths might lead to extensors~$\phi(P)$ that cancel in the sum in~\eqref{eq:walkpol = pathpol}.
  Let us consider a case where this never happens by assuming that the graph~$G$ has at most one $k$-path.
  Then the sum over paths in~\eqref{eq:walkpol = pathpol} has at most one term and cancellations cannot occur.

This allows us to establish Thm.~\ref{thm decide sub} for the special case where $H$ is the $k$-path and the number~$C$ of occurrences of $H$ in $G$ is either zero or one.

\begin{algor}{U}{Detect unambiguous $k$-path.}{%
  Given directed graph~$G$ and integer $k$, such that the number of $k$-paths in $G$ is $0$ or $1$, this algorithm determines if $G$ contains a $k$-path.}
  \item [U1] (Set up $\phi$.) Let $F=\QQ$. 
    Let~$\phi$ be the Vandermonde extensor-coding as in~\eqref{eq:Vandermonde extensor}.
  \item [U2] (Compute the walk-sum)
    Compute $\walkpol(G;\phi)$ as in \eqref{eq:abstract adjacency matrix}.
  \item [U3] (Decide.)
    If $\walkpol(G;\phi)$ is non-zero, then return `yes.' 
    Otherwise, return `no.' \qed
\end{algor}

\begin{thm} \label{thm decide path}
  Algorithm U is a deterministic algorithm for the unambiguous $k$-path problem with running time $2^k (n+m)\poly(k)$.
\end{thm}
\begin{proof}
  Consider the extensor $\walkpol(G;\phi)$ computed in Step~U2.
  If~$G$ contains no $k$-path, then $\walkpol(G;\phi)=0$ holds by~\eqref{eq:walkpol = pathpol}.
  Otherwise, we have~$\walkpol(G;\phi)=\phi(P)$ for the unambiguous $k$-path~$P$ in $G$.
  Let $P=v_{i_1}\ldots v_{i_k}$.
  By our choice of $\phi$ in U1, Lemma~\ref{lem: vandermonde extensor} implies $\walkpol(G;\phi)=d\cdot \cano {[k]}$ with~$d\ne 0$.

  The running time of Algorithm~U is clearly dominated by U2.
  As we discussed in Sec.~\ref{sec: rep and comp}, the value $f(G,\phi)$ can be computed with~$k\cdot O(n+m)$ operations in $\Lambda(F^k)$, each of which can be done with~$O(k2^k)$ operations in~$F$.
  The Vandermonde extensor-coding~$\phi$ uses only integer vectors and the absolute value of $f(G,\phi)$ is bounded by~$n^{\poly(k)}$.
  In the usual word-RAM model of computation with words in~$\set{-n,\dots,+n}$, we can thus store each number using~$\poly(k)$ words.
  We conclude that Algorithm~U has the claimed running time.
\end{proof}

\subsection{Blades and Lifts}

The reason that cancellations can occur in~\eqref{eq:walkpol = pathpol} is that the coefficients~$d\in F$ in \eqref{eq: d for xi} may be negative.
We will now give a general way to modify an extensor-coding in such a way that these coefficients become~$d^2$ and thus are always positive.

Instead of~$\Lambda(F^k)$, we will now work over $\Lambda(F^{2k})$.
For an extensor $x=\sum_{i\in\set{1,\dots,k}} a_i\cano i\in F^k\subseteq \Lambda(F^k)$, we define its \emph{lifted} version~${\overline x\in\Lambda^2(F^{2k})}$ as the blade
\begin{equation}\label{eq: tensor square embedding}
\overline x
=
\paren[\Big]{
\sum_{i\in\set{1,\dots,k}} a_i\cano i
}
\land
\paren[\Big]{
  \sum_{j\in\set{1,\dots,k}} a_j\cano {j+k}
}
\,.
\end{equation}
If we let $0\in F^k$ denote the zero vector in $F^k$, we can write this as \[\overline{x}= \begin{pmatrix}x\\0\end{pmatrix}\wedge\begin{pmatrix}0\\x\end{pmatrix}\,.\]
  Crucially, every $\overline{x}$ is decomposable, so Lemma~\ref{lem: main} applies. 

  For an extensor-coding~$\xi\colon V(G)\to F^k$, we define the lifted extensor-coding~$\overline\xi\colon V(G)\to\Lambda(F^{2k})$ by setting $\overline\xi(v)=\overline{\xi(v)}$ for all~${v\in V(G)}$.
For a path~$P\in\mathscr P$, with $P=w_{1}\cdots w_{k}$, the correspondence between~$\xi(P)$ and $\overline{\xi}(P)$ is as follows.
Consider the $k\times k$ matrix $\Xi_P$ of extensors given by
\[ \Xi_P = \begin{pmatrix} \xi(w_1)\hdots \xi(w_k)\end{pmatrix}\,. \]
  From Property \ref{wedge determinant}, we get \[\xi(P) = (\det \Xi_P) \cano {[k]}\,,\]
and
\[
    \overline{\xi} (P) = 
  \det \begin{pmatrix} 
  \xi(w_1) & 0        & \hdots & \xi(w_k) & 0 \\
  0        & \xi(w_1) & \hdots & 0        & \xi(w_k)
  \end{pmatrix}
\cano {[2k]}\,.\]
Using basic properties of the determinant, we can rewrite the coefficient of $\cano {[2k]}$ to
\begin{multline*}
(-1)^{\binom{k}{2}}
 \det \begin{pmatrix} 
  \xi(w_1) &  \hdots & \xi(w_k) & 0         & \hdots & 0  \\
  0        &  \hdots & 0        & \xi(w_1)  & \hdots & \xi(w_k)
\end{pmatrix}  = \\
(-1)^{\binom{k}{2}}
(\det \Xi_P) \cdot (\det \Xi_P)  =
(-1)^{\binom{k}{2}} (\det \Xi_P)^2 \,.
\end{multline*}
Thus, we have
\[ \overline{\xi}(P) =\pm (\det \Xi_P)^2 \cano {[2k]}\,,\]
where the sign depends only on $k$.

We evaluate the walk-sum over $\Lambda(F^{2k})$ at $\overline\xi$ to obtain:
\begin{equation}\label{eq: detsquared}
  \walkpol(G;\overline\xi)
  =\pm
  \sum_{P\in\mathscr P} (\det \Xi_P)^2
  \cdot
  \cano {[2k]}
\,.
\end{equation}

\subsection{Deterministic Algorithm for Path Detection}
\label{sec: deterministic alg}%
As an application of the lifted extensor-coding, let~$\phi\colon V(G)\to F^k$ be the Vandermonde extensor-coding from Lemma~\ref{lem: vandermonde extensor}.
We imitate Algorithm~U to arrive at a deterministic algorithm for $k$-path.
Our algorithm slightly improves upon the time bound of $4^{k+o(k)}\cdot \poly(n)$ of Chen \emph{et al.}~\cite{DBLP:conf/soda/ChenLSZ07,doi:10.1137/080716475},
but does not come close to the record bound $2.5961^k \cdot \poly(n)$ of Zehavi~\cite{Zehavi14}.
\begin{thm}[Superseded by \cite{Zehavi14}]
There is a deterministic algorithm that, given a directed graph $G$,
checks if~$G$ has a path of length~$k$ in time $4^k (n+m)\poly(k)$.
\end{thm}
\begin{proof}
  The algorithm is just Algorithm~U, except that we evaluate the walk-sum over~$\Lambda(F^{2k})$ and at~$\overline\phi$.
  The correctness of this algorithm follows from~\eqref{eq: detsquared}.
  Each addition~$y+z$ in~$\Lambda(F^{2k})$ can be carried out using~$O(2^{2k})$ addition operations in~$F$, and each multiplication~$y\wedge\overline x$ with elements of the form~$\overline x$ for~$x\in F^k$ takes at most~$O(2^{2k}k^2)$ operations in~$F$, as discussed in~Sec.~\ref{sec: rep and comp}.
  Overall, this leads to the claimed running time.
\end{proof}

\newcommand{\E}{\operatornamewithlimits{\mathbf{E}}}
\newcommand{\Var}{\operatorname{Var}}
\newcommand{\Cov}{\operatorname{Cov}}

\subsection{Bernoulli Vectors}
\label{sec: Bernoulli coding}

We present our algorithm for approximate counting.
Now instead of the Vandermonde extensor-coding as in Lemma~\ref{lem: vandermonde extensor}, we sample an extensor-coding~$\beta\colon V(G)\to \set{-1,1}^k$ uniformly at random.

The approximate counting algorithm is based on the following observation: 
If $B_P$ is the $k\times k$ matrix corresponding to $\beta(w_1)$, $\ldots$, $\beta(w_k)$, then
all matrices~$B_P$ are sampled from the same distribution.
Thus, the random variables $\det B_P^2$ have the same mean~${\mu>0}$.
The expectation of the sum of determinant squares is~$\mu\cdot\abs{\paths}$, from which we can recover an estimate for the number of paths.
Our technical challenge is to bound the variance of the random variable~$\det B_P^2$.

\begin{algor}{C}{Randomized counting of $k$-path.}{%
  Given directed graph $G$ and integers $k$ and $t$, approximately counts the number of $k$-paths using $t$ trials.}
\item [C1] (Initialize.) Set $j=1$.
\item [C2] (Set up $j$th trial.)
    For each $i\in\{1,\ldots,n\}$,
    let $\beta(v_i)$ be a column vector of $k$ values chosen from $\pm 1$ independently and uniformly at random.
  \item [C3] (Compute scaled approximate mean $X_j$.) 
    Compute~$X_j$ with $\walkpol(G;\overline\beta)= X_j\cdot \cano {[2k]}$.
  \item [C4] (Repeat $t$ times.) If $j<t$ then increment $j$ and go to C2.
  \item [C5] (Return normalized average.)
    Return $(X_1+\cdots +X_t)/(k! t)$
\end{algor}

We are ready for the special case of Theorem~\ref{thm approx count sub}, approximating $\Sub(H,G)$ when $H$ is the $k$-path.
In this case, $\Sub(H,G)= |\paths|$.

\begin{thm}\label{thm: algorithm C}
  For any $\varepsilon>0$, Algorithm~C
  produces in time $(4^k/ \varepsilon^2)\cdot (n+m) \cdot \poly(k)$ 
  a value $X$ such that with probability at least~$99\%$, we have
\[
  (1-\varepsilon) \cdot  |\paths|\leq X \leq (1+\varepsilon)\cdot |\paths|\,.
\]
\end{thm}

A matrix whose entries are i.i.d.\ random variables taking the values $+1$ and $-1$ with equal probability $\frac{1}{2}$ is called \emph{Bernoulli}. 
We need a result from the literature about the higher moments of the determinant of such a matrix.

\begin{thm}[\cite{Nyquist54}] \label{thm:randmat}
Let $B$ be a $k\times k$ Bernoulli matrix.
Then,
\begin{align}
  \E\det B^2 &= k! \label{eq:secmom}\\
\E \det B^4 &\leq (k!)^2 \cdot k^3 \,. \label{eq:fourthmom}
\end{align}
\end{thm}

For completeness, we include a careful proof for a slightly different distribution in Appendix~\ref{sec: random matrix}.

\begin{proof}[Proof of Theorem~\ref{thm: algorithm C}]
  Run algorithm C with $t=100 k^3/\varepsilon^{2}$.
  Set $\mu = |\paths|$.
  Recall from~\eqref{eq: detsquared} that $X_j$ can be written as
  \begin{equation}\label{eq:X_j}
    X_j =\pm \paren{ \det B_1^2 + \det B_2^2 +\cdots +\det B_\mu^2 }\,,
\end{equation}
where for $i\in \{1,\ldots,\mu\}$, each $B_i$ is a submatrix of of the $k\times n$ matrix with columns  $ \beta(v_1)$, $\beta(v_2)$,$\cdots$, $\beta(v_n)$.
The sign can be easily computed and only depends on~$k$; we assume without loss of generality that it is $+1$.
  By our choice of $\beta$ in Step~C2, each $B_i$ is therefore a Bernoulli matrix, but they are not independent.

  By Theorem~\ref{thm:randmat}, we have $\E \det B_i^2 = k!$ for each $i\in\{1,\ldots, \mu\}$, so by linearity of expectation,
  \[ \E X_j = \mu k!\,.\]
 
  We turn to $\Var X_j$, which requires a bit more attention.
  For all $i,\ell\in \set{1,\ldots, \mu}$, the matrices $B_i$ and $B_\ell$ 
  follow the same distribution, so $\Var \det B_i^2 = \Var \det B_\ell^2$.
  Thus, using Cauchy--Schwartz, we have
  \begin{multline*} \Cov (\det B_i^2,\det B_\ell^2) =  
    \sqrt{(\Var \det B_i^2)\cdot(\Var \det B_\ell^2)} =\\ \sqrt{(\Var \det B_i^2)^2} = \Var\det B_i^2
  \leq \E \det B_i^4
    \leq (k!)^2 k^3
    \,,
  \end{multline*}
  where the last two inequalities uses $\Var Y \leq \E Y^2$ with $Y = \det B_i^2$ and \eqref{eq:fourthmom} in Theorem~\ref{thm:randmat} with $B= B_i$.
  We obtain
  \begin{multline*}
    \Var X_j = 
    \Cov(X_j, X_j) = 
    \Cov\biggl(\sum_{i=1}^\mu \det B_i^2, \sum_{\ell=1}^\mu \det B_\ell^2\biggr) =\\
    \sum_{i,\ell = 1}^\mu \Cov(\det B_i^2, \det B_\ell^2) \leq 
    \mu^2 \cdot (k!)^2 \cdot k^3\,.
  \end{multline*}
  	
\medskip
  Now consider the value $X$ returned by the algorithm in Step~C5 and observe $X=(X_1+\ldots+X_t)/(k!t)$.
  By linearity of expectation, we have $\E X = t\mu k!/(k!t) = \mu$.
  Recalling that $\var{a \cdot X} = a^2 \cdot \var{X}$ for a random variable $X$ and a scalar $a$, by independence of the $X_j$, we have
  \begin{equation*} \Var X = 
    \Var \biggl(\frac{1}{k!t} \sum_{j=1}^t X_j \biggr) = 
    \frac{1}{(k!t)^2} \sum_{j=1}^t\Var X_j \leq
    \frac{1}{(k!t)^2} t \mu^2 (k!)^2 k^3 =
    \frac{\mu^2  k^3}{t}\,.
  \end{equation*}
Now Chebyshev's inequality gives 
\[
\Pr(| X - \mu| \geq \varepsilon  \mu) \leq
\frac{\Var X}{\varepsilon^2 \mu^2 } \leq
\frac{\mu^2  k^3}{\varepsilon^2  \mu^2 t } =
\frac{1}{100}\,,
\]
which implies the stated bound.

The claim on the running time follows from the discussion in Sec. \ref{sec: rep and comp} and the representation of the input as adjacency lists.
\end{proof}

\subsection{Edge-Variables}
\label{sec: Edge-Variables}

We extend Algorithm~U from the unambiguous case to the case where the number of $k$-paths is bounded by some integer $C$.
The construction uses a coding with formal variables on the edges.
To this end, enumerate $E$ as $\{e_1,\ldots, e_m\}$ and introduce the set $Y$ of formal variables $\{y_1,\ldots, y_m\}$.
Our coding maps $e_j$ to $y_j$.

We then use the following theorem about deterministic polynomial identity testing of sparse polynomials due to Bläser \emph{et al.}:

\begin{thm}[Theorem 2 in \cite{DBLP:journals/ipl/BlaserHLV09}] 
\label{thm: pit}
  Let $f$ be an $m$-variate polynomial of degree $k$ consisting of $C$ distinct monomials with integer coefficients, with the largest appearing coefficient bounded in absolute value by $H$.
  There is a deterministic algorithm which, given an arithmetic circuit of size $s$ representing $f$, 
  decides whether $f$ is identically zero in time $O((mC \log k)^2 s \log H)$
\end{thm}

To use this result, we need to interpret the walk-sum as a small circuit in the variables $Y$ with integer coefficients.
This requires `hard-wiring' every skew product in the exterior algebra by the corresponding small circuit over the integers. 
Algorithm~F contains a detailed description.

\begin{algor}{F}{Detect few $k$-paths}{%
Given directed graph~$G$ and integer~$k$, such that the number of $k$-paths in $G$ is at most~$C$, this algorithm determines if $G$ contains a $k$-path.}
\item[F1] [Set up $\eta$.] 
  Let $F=\ZZ$ and define $\eta\colon V(G)\cup E(G)\rightarrow \Lambda(F^k)[Y]$ by $\eta(v)=\phi(v)$ and $\eta(e_j)=y_j$.
\item[F2] [Circuit $K$ over $\Lambda(F^k)[Y]$.] 
  Let $K$ be the skew arithmetic circuit from~\eqref{eq: walk-sum matrixpower} for computing $\walkpol(G;\eta)$ from its input gates labeled by~$\eta(v)$ for $v\in V(G)$ and $\eta(e)$ for $e\in E(G)$.
\item[F3] [Circuit $L$ over $\ZZ[Y]$.]
  Create a circuit $L$ with inputs from $\ZZ$ and~$Y$ as follows. 
  Every gate~$g$ in $K$ corresponds to $2^k$ gates~$g_I $ for $I\subseteq\{1,\ldots, k\}$ such that $g=\sum_I g_I\cdot \cano I$. 
  When~$g$ is an input gate of the form $g=\phi(v_i)$ the only nonzero gates in $L$ are $g_{\{j\}} =i^j$, an integer. 
  When~$g$ is an input gate of the form $g=y_j$ then the only nonzero gate is the variable  $g_{\emptyset} = y_j$.
  If $g=g'+g''$ then $g_I$ is the addition gate computing $g'_I + g''_I$.
  If $g$ is the skew product~$g'\cdot g''$, where $g''$ is an input gate, then $g_I$ is the output gate of a small subcircuit that computes
  \[
    \sum_{\substack{J\subseteq I\\\abs{J}\leq1}} (-1)^{\sgn(I\setminus J,J)} g'_{I\setminus J} g''_J\,.
  \]
  (This is \eqref{eq: alternating subset convolution}, noting $g''_J=0$ for $|J|>1$.)
  If $g$ is the output gate of $K$ then $g_{\{1,\ldots,k\}}$ is the output gate of $L$.
\item[F4] [Decide.]
  Use the algorithm from the above theorem to determine if $L$ computes the zero polynomial.
  Return that answer.
\end{algor}

We are ready to establish Theorem~\ref{thm decide sub} for the case where the pattern graph $H$ is a path.

\begin{thm} \label{thm: algorithm F}
Algorithm $F$ is a deterministic algorithm for the $k$-path problem when there are at most $C \in \NN$ of them, and runs in time $C^2 2^k n^{O(1)}$.
\end{thm}
\begin{proof}
  Let $G$ be a graph with at most $C$ paths of length $k$.
  First, we argue for correctness of Algorithm $F$. 
  From \eqref{eq: walk-sum def}, it follows that the circuit $K$ outputs 
  \[
    f(G;\eta) = 
    \sum_{P \in \paths} \left( \prod_{e_i \in P} y_i \right) \cdot \det(\Phi_P) \cdot \cano{[k]} \in \Lambda(F^k)[Y] \,,
  \]
  where $\Phi_P$ is the Vandermonde matrix associated with the vertices on $P$ from \eqref{eq: Vandermonde matrix}.
  By the construction of $L$, the output gate of $L$ computes the polynomial
  \[
    \sum_{P \in \paths} \left( \prod_{e_i \in P} y_i \right) \cdot \det(\Phi_P) \in F[Y] \,,
  \]
  which is just an $m$-variate, multilinear polynomial over the integers.
  Note that, by construction, all the appearing determinants are non-zero.
  Since all our graphs are directed, any path is already uniquely determined by the unordered set of edges that appear on it.
  It follows that the monomials belonging to the distinct $k$-paths in a graph, each formed as the product of the edge variables corresponding to the edges on the path, are linearly independent.
  Therefore, the monomials of the polynomial in $Y$ computed by $L$ are in bijective correspondence with the $k$-paths in $G$. 
  Theorem \ref{thm: pit} thus yields the correct answer. 

  As for the running time, we see that every gate in $K$ is replaced by at most $2^k(k+1)$ new gates to produce $L$.
  Since $K$ was of size $O(k(n+m))$, the resulting circuit $L$ is of size $O(2^k (n+m)\poly (k))$ and can be constructed in this time. 
  Since, as noted, the monomials in the polynomial computed by $L$ are in bijection with the $k$-paths in $G$, there are at most $C$ many.
  The application of Theorem \ref{thm: pit} is thus within the claimed running time bound. 
\end{proof}

\section{Connection to Previous Work}
\label{sec: Previous}
In this section, we show how our approach using exterior algebras specializes to the group algebra approach of Koutis ~\cite{DBLP:conf/icalp/Koutis08} when the ground field has characteristic two.
We also argue that the combinatorial approach of Björklund \emph{et al.}~\cite{DBLP:journals/jcss/BjorklundHKK17} using \emph{labeled walks} can be seen as an evaluation over an exterior algebra.
Moreover, we show how color-coding~\cite{DBLP:journals/jacm/AlonYZ95} arises as a special case,
and present the recent approach of representative paths due to Fomin et al.~\cite{DBLP:journals/jacm/FominLPS16} in the language of exterior algebra.

\subsection{Random Edge-Weights}
\label{sec: Random Edge-Weights}

We begin with a randomized algorithm for detecting  a $k$-path in a directed graph, recovering Koutis's and Williams's result.

\begin{thm}[\cite{DBLP:conf/icalp/Koutis08, DBLP:journals/ipl/Williams09}]
  There is a randomized algorithm for the $k$-path problem with running time $2^k(n+m)\poly(k)$.
\end{thm}
\begin{proof}
  The algorithm is the baseline Algorithm~U, but with the following step replacing U1:
  \begin{description}[noitemsep,labelindent=0pt,labelwidth=2.2em,labelsep=0pt,leftmargin=!]%
    \item[U1\boldmath$'$]
      Enumerate the edges as $E=\{e_1,\ldots, e_m\}$ and choose $m$ integers $r_1,\ldots, r_m\in\{1,\ldots,100k\}$ uniformly at random.
	Define the extensor-coding $\rho$ on $V(G)\cup E(G)$ by 
	\[ v_i\mapsto \phi(v_i), \qquad e_j\mapsto r_j\,.\]
  \end{description}
  The rest is the same, with $\rho$ instead of $\phi$. 

  The correctness argument is a routine application of polynomial identity testing:
  The expression $f(G;\rho)$ can be understood as the result of the following random process.
  Introduce a formal `edge' variable $y_e$ for each $e\in E$ and consider the expression
  \begin{equation}\label{eq: walk poly}
    \sum_{w_1\cdots w_k\in \mathscr P} y_{w_1w_2}\cdots y_{w_{k-1}w_k}\cdot \phi(w_1\ldots w_k)
  \end{equation}
  as a polynomial of degree $k$ in the variables $y_{e_1},\ldots, y_{e_m}$.
  In a directed graph, every path is uniquely determined by its set of (directed) edges.
  Thus, if $\mathscr P \neq \emptyset$ then \eqref{eq: walk poly} is a nonzero polynomial.
  The walk-sum $f(G;\rho)$ is an evaluation of this polynomial at a random point $y_{e_1}=r_1,\ldots,y_{e_m}=r_m$.
  By the DeMillo--Lipton--Schwartz--Zippel Lemma, $f(G;\rho)$ is nonzero with probability $\tfrac{1}{100}$.
\end{proof}

\subsection{Group Algebras}
\label{sec: Group Algebras}
Let $R$ be a ring and let $M$ be a monoid with multiplication $\ast$.
We denote with $R[M]$ the \emph{monoid algebra of $M$ over $R$}.
If $M$ is actually a group, we call $R[M]$ the \emph{group algebra of $M$ over~$R$}.
That is, $R[M]$ is the set of all finite formal linear combinations of elements from~$M$ with coefficients in $R$.
An element of $R[M]$ is thus of the form $\sum_{m \in M} r_m \cdot m$, with only finitely many of the $r_m\in R$ non-zero.
Elements from $R[M]$ admit a natural point-wise addition and scalar multiplication.
Multiplication in $R[M]$, written $\bullet$, 
is defined by the distributive law,
\[
\left(\sum_{m\in M} c_m \cdot m\right)\bullet \left(\sum_{m \in M} d_m \cdot m\right)
=
\left(\sum_{g,h \in G} (c_g \cdot d_{h}) \cdot (g\ast h) \right) \,,
\]
which is again an element of $R[M]$.

\medskip
As the name suggests, the monoid algebra $R[M]$ is indeed an $R$-algebra, and is of dimension~$|M|$.
Usually, multiplication and addition in the ground ring $R$, the monoid $M$, and the group algebra~$R[M]$ are all denoted by $\cdot$ and $+$.

\begin{prop}
Let $F$ be of characteristic two and $F^k$ the free vector space of dimension $k$ with basis $\{\cano 1,\ldots,\cano k\}$.
Then, the group algebra $F[\ZZ_2^k]$ is isomorphic to $\extalg{F^k}$.
\end{prop}
\begin{proof}
  We denote with $\cano i \in \ZZ_2^k$ for $i\in\set{1,\dots,k}$ the $i$th unit vector.
  The morphism induced by mapping $\extalg{F^k} \ni \cano i \mapsto (1 + \cano i) \in F[\ZZ_2^k]$ is an isomorphism.
\end{proof}

\begin{rem}
The previous proposition shows that over fields of characteristic two, 
our exterior algebras specialize exactly to the group algebras used by Koutis and Williams~\cite{DBLP:conf/icalp/Koutis08,DBLP:journals/ipl/Williams09},
and therefore, the approach of using random edge-weights in the coding $\rho$ from Section \ref{sec: Random Edge-Weights} specializes to Williams' algorithm~\cite{DBLP:journals/ipl/Williams09} over fields of characteristic two and sufficient size, albeit with deterministically chosen vectors at the vertices, which of course also could be done randomly without changing anything about the result.
\end{rem}

\subsubsection*{Exterior Algebras as Quotients of Monoid Algebras}
We have seen that the above group algebras are exterior algebras in characteristic two,
and now consider the other direction.
For $k \in \NN$, consider the free monoid $E^\ast$ over the generators $E := \{\cano 1,\ldots,\cano k,\mu,\theta\}$, and impose these relations on $E^\ast$:
The element $\theta$ is a zero, \emph{i.e.}, $\theta x = x\theta = \theta$ for all $x \in E^\ast$,
and $\mu$ central, \emph{i.e.}, $\mu x = x \mu$ for all $x \in E^\ast$, 
and we shall have for all $i$ that $\cano i^2 = \theta$. 
We further demand that $\cano i \cano j = \mu \cano j \cano i$ and $\mu^2 = 1_E$ hold.
Let $S$ be the quotient of $E^\ast$ by these relations, and consider $F[S]$. 
Let $I_S$ be the ideal generated by $\{\theta, \mu + 1\}$.
Naturally in $F[S]/I_S$, we have $\theta = 0$ and $\mu = -1$, and hence $\cano i^2 = 0$ and $\cano i \cano j = -\cano j \cano i$.
Thus, $F[S]/I_S$ is \emph{precisely} the exterior algebra over $F^k$.
Hence, not only are the above-mentioned group algebras a special case of an exterior algebra, 
but any exterior algebra arises as a quotient of some monoid algebra.
Note that this representation of exterior algebras (and, more generally, Clifford algebras) as quotients of certain monoid (or group) algebras is folklore.

\subsection{Labeled Walks}
\label{sec: Labeled Walks}
The main goal of the labeled walk approach of Björklund \emph{et al.}~\cite{DBLP:journals/jcss/BjorklundHKK17}  was to give an algorithm for the \emph{undirected} case running in time $1.66^k\cdot \poly(n)$.
This is achieved by a method called \emph{narrow sieves}, which involves reducing the number of so-called labels used on the graph.
The underlying walk labeling idea itself, however, remains valid also on directed graphs and when keeping \emph{all} labels, and then reproduces the randomized $2^k\cdot \poly(n)$ runtime bound of Williams~\cite{DBLP:journals/ipl/Williams09}.
This is nicely laid out in the textbook by Cygan \emph{et al.}~\cite[Section 10.4]{DBLP:books/sp/CyganFKLMPPS15},
and the following presentation is guided by theirs. 

Consider now $\lambda,$ the extensor-coding for labeled walks. 
That is, let $y_e$ be variables associated with each directed edge $e \in E(G)$,
and let a vector of variables $(x_i^{(1)},\ldots,x_i^{(k)})^T$ be associated with each vertex $v_i\in V(G)$.
The superscript index is referred to as the \emph{label} of a vertex in a walk.
Consider the following polynomial in the $y_e$ and $x^{(j)}_i$:
\begin{align}
  P(x,y) = \sum_{w_1\cdots w_k\in\walks} \sum_{\ell \in S_k} \prod_{i=1}^k y_{w_iw_{i+1}} \prod_{i=1}^k x^{\ell(i)}_i.
\end{align}
The crucial insight is that over characteristic $2$, the sum can be restricted to paths instead of walks:
\begin{align}
  P(x,y) = \sum_{w_1\cdots w_k\in\paths} \sum_{\ell \in S_k} \prod_{i=1}^k y_{w_iw_{i+1}} \prod_{i=1}^k x^{\ell(i)}_i.
\end{align}
In this form, the statement is true only over characteristic two.
However, even over characteristic~$0$, a similar statement can be made when taking into account the sign of the permutation $\ell$.
We may now observe that the inner sum is just a determinant of a suitably chosen matrix,
namely the $k\times k$ matrix $X(w_1\cdots w_k) := (x^{(i)}_{w_j})_{i,j}$ indexed by pairs of numbers and vertices,
and we can write
\begin{align*}
  P(x,y) = \sum_{w_1\cdots w_k\in\paths}  \prod_{i=1}^k y_{w_iw_{i+1}} \det(X(w_1\cdots w_k))\,.
\end{align*}
Here, the Matrix $X(P)$ for a path $P \in \paths$ plays precisely the r\^{o}le that the matrix $\Phi_P$ played in the proof of Theorem \ref{thm: algorithm F}, just that this time, it carries variable entries.

It is now easy to see that this is once again just the evaluation of the circuit computing the $k$-walk extensor over characteristic two by the property of the wedge product expressed in Equation \eqref{eq: wedge as det}.
In short, the walk-sum $f(G;\lambda)$ for the extensor-coding $\lambda$ achieves
\[
f(G;\lambda) = P(x,y)
\]
whenever $F$ is of characteristic two.

This gives the connection between $\lambda$ and $\rho$ from Section \ref{sec: Random Edge-Weights} over characteristic two, and by the remark in Section \ref{sec: Group Algebras}, also the connection between the group-algebra approach, identifying the three techniques as one. 

\subsection{Color-coding}
\label{sec: color-coding}
Let us see how the color-coding-technique by Alon, Yuster, and Zwick~\cite{DBLP:journals/jacm/AlonYZ95} arises as a special case of extensor-coding.
Consider a coding with basis vectors of $F^k$ taken uniformly and independently,
\[\chi(v) \in\{e_1,\ldots, e_k\}\,.\]
Readers familiar with \cite{DBLP:journals/jacm/AlonYZ95} are encouraged to think of the basis vectors as $k$ colors.
Note that if $P=w_1\cdots w_k$ is a path then its walk extensor $\chi(P)$ vanishes exactly if the $k\times k$-matrix whose columns are the random unit vectors $\chi(w_1)\cdots\chi(w_k)$ is singular. 
Thus, \[
\Pr (\chi(P) = 0) = \frac{k!}{k^k}\leq e^{-k} \,.\]
We lift $\chi$ to $\overline\chi\colon V(G)\rightarrow \Lambda^2(F^{2k})$, ensuring $\overline\chi(P) = \{0\cdot e_{[2k]}, 1\cdot e_{[2k]}\}$, to avoid cancellation.

Let us write $Z(F^k)$ for the subalgebra of $\Lambda(F^{2k})$ generated by $\{\overline{\cano 1},\ldots ,\overline{\cano k}\}$, called the \emph{Zeon-algebra}. It already made an appearance in graph algorithms in the work of Schott and Staples~\cite{DBLP:journals/cma/SchottS11}.

\begin{lem} 
\label{lem: zeons}
$Z(F^k)$ is commutative and of dimension $2^k$, and its generators $\overline{\cano i}$ square to zero.
Furthermore, addition and multiplication can be performed in $2^k\cdot \poly(n)$ field operations.
\end{lem}
\begin{proof}
Directly from the definition of the exterior algebra, $\overline{\cano i} \land \overline{\cano i} = 0$.
Furthermore,
\begin{multline*}
\overline{\cano i} \wedge \overline{\cano j} =
\cano{i} \wedge \cano{i+k} \wedge \cano{j} \wedge \cano{j+k} = -\cano{i} \wedge \cano{j} \wedge \cano{i+k} \wedge \cano{j+k} = 
 \cano{j} \wedge \cano{i} \wedge\cano{i+k} \wedge\cano{j+k} =\\ -\cano{j} \wedge \cano{i}\wedge\cano{j+k} \wedge\cano{i+k} =
 \cano{j} \wedge\cano{j+k} \wedge \cano{i} \wedge\cano{i+k}  =  \overline{\cano{j}} \wedge \overline{{\cano i}}\,,
\end{multline*}
and therefore $Z(F^k)$ is commutative.
It is readily verified that the elements $\overline{\cano I}$ with $I \subseteq [k]$ form a basis of $Z(F^k)$.
By renaming $\overline{\cano i}$ as, say, $X_i$, we recognize $Z(F^k)$ as the $F$-algebra of multilinear polynomials in variables $X_i, 1 \leq i \leq k$ with the relations $X_i^2 = 0$ for all $1 \leq i \leq k$.
Addition is performed component-wise and can be done trivially in the required bound.
By standard methods, such as Kronecker substitution and Schönhage--Strassen-multiplication (see, \emph{e.g.}, \cite{DBLP:books/daglib/0031325}), or more directly, fast subset convolution \cite{DBLP:conf/stoc/BjorklundHKK07}, multiplication of multilinear polynomials modulo $X_i^2$ can be performed in $Z(F^k)$ in the required time bound.
\end{proof}

Thus, we can evaluate the walk-sum $f(G;\overline\chi)$ in time $2^k(n+m)\poly(k)$.
Repeating the algorithm $\mathrm e^k$ times we arrive at time $(2\mathrm e)^k(n+m)\poly(k)$, as in \cite{DBLP:journals/jacm/AlonYZ95}.

\medskip
We apply these constructions in a similar setting in Appendix \ref{sec: mult det}.

\subsection{Representative Paths}

The idea to represent the $\Omega(n^kk!)$ many $k$-paths in $G$ by a family of only $f(k)\poly(n)$ many combinatorial objects goes back to the original $k$-path algorithm of Monien~\cite{MONIEN1985239}.
Recent representative-sets algorithms \cite{DBLP:journals/jacm/FominLPS16,Zehavi14}, including the fastest deterministic $k$-path algorithms, follow this approach, maintaining representative families of subsets of a linear matroid.

One of those constructions is inspired by Lovász's proof of the Two-Families theorem, which is originally expressed in exterior algebra \cite{lovasz77}.
In fact, as pointed out by Marx \cite[Proof of Lemma 4.2]{DBLP:journals/tcs/Marx09}, the column vectors in the matroid representation are exactly Vandermonde extensors.
Fomin \emph{et al.} \cite[Theorem 1]{DBLP:journals/jacm/FominLPS16} develop this idea in detail for $k$-path, but their presentation abandons the exterior algebra and continues in the framework of uniform matroids.

We can give a relatively short and complete presentation.
This has only expository value; the time bounds in this construction are not competitive.

\medskip
For a set $\mathscr R$ of walks and an extensor coding $\xi$ to $\Lambda(F^k)$ we define the \emph{extensor span} $\langle \mathscr R\rangle$ via
\[\langle\mathscr R\rangle = \operatorname{span}\paren[\Big]{\setc{\xi(R)}{R\in \mathscr R}}\,,\]
that is, $\langle \mathscr R\rangle$ is the set of linear combinations over $F$ of extensors viewed as $2^k$-dimensional vectors.
A~set $\mathscr R$ of walks \emph{represents} another set $\mathscr P$ of walks if $\xi(P)\in\langle\mathscr R\rangle$ for all $P\in \mathscr P$.
For $p\in\{1,\ldots, k\}$, we write $\mathscr P_v^p$ for the set of length-$p$ paths of $G$ that end in~$v$.
We will use the Vandermonde coding $\phi$ for $\xi$.

\begin{algor}{R}{Detect $k$-paths using representative paths.}{%
  Given directed graph $G$ and integer $k$, this algorithm determines if $G$ contains a $k$-path.
  For each $p\in\set{1,\ldots, k}$ and $v\in V(G)$, the algorithm computes a set $\mathscr R_v^p$ of paths such that
  \begin{equation}\label{eq: inv rep 1}
    \phi(P)\in \langle \mathscr R_v^p\rangle\qquad 
      \text{for each $P\in \mathscr P_v^p$}
  \end{equation}
  and 
  \begin{equation}\label{eq: inv size 1}
    \abs{\mathscr R_v^p} \leq 2^k\,.
  \end{equation}
}
  \item [R1] (First round.)
    Let $p=1$.
    For each $v\in V(G)$, set $\mathscr R^1_v= \set{v}$, the singleton set of $1$-paths.
  \item [R2] (Construct many representative walks.)
    For each $v\in V(G)$, set
    \begin{equation}\label{eq: Q def}
      \mathscr Q_v^{p+1}  = \setc[\big]{Rv}{R\in \mathscr R_u^p\text{ and } uv\in E(G)}\,.
    \end{equation}
  \item [R3] (Remove redundant walks.)
    For each $v\in V(G)$, set $\mathscr R_v^{p+1} = \emptyset$.
    For each $Q\in \mathscr Q_v^{p+1}$ in arbitrary order, if 
    \( \phi(Q)\notin \langle \mathscr R_v^{p+1}\rangle \,,\)
    then add $Q$ to $\mathscr R_v^{p+1}$.
    [Now $\phi(Q)\in \langle\mathscr R_v^{p+1}\rangle$.]
  \item [R4] (Done?)
    If $p+1<k$ then increment $p$ and go to R3.
    Otherwise return `true' if and only if $\mathscr R_v^k\neq \emptyset$ holds for some $v\in V(G)$.
\end{algor}

\begin{prop}[Theorem~1 in \cite{DBLP:journals/jacm/FominLPS16}]
  Algorithm R is a deterministic algorithm for $k$-path with running time $\exp(O(k))\poly(n)$.
\end{prop}

\begin{proof}
  We need to convince ourselves that the invariants \eqref{eq: inv rep 1} and \eqref{eq: inv size 1} hold, and that the constructed sets $\mathscr R_v^p$ only contain paths from~$\mathscr P_v^p$.
  For the size invariant \eqref{eq: inv size 1}, it suffices to observe that $\langle \mathscr R_v^{p+1}\rangle$ is a subspace of $\Lambda(F^k)$ and thus has dimension at most $2^k$.
  Each element $Q$ was added in Step~R3 only if it increased the dimension of $\langle \mathscr R_v^{p+1}\rangle$, which can happen at most~$2^k$ times.

  We prove~\eqref{eq: inv rep 1} by induction on~$p$.
  For $p=1$, we have $\mathscr P_v^1=\mathscr R_v^1$ for all $v\in V(G)$ by Step~R1.
  For the induction step, assume that $p$ satisfies~$\phi(\mathscr P_u^p)\subseteq \langle\mathscr R_u^p\rangle$ for all $u\in V(G)$.
  Let $v\in V(G)$ and consider a path~$Puv$ from $\mathscr P_v^{p+1}$.
  To establish the inductive claim, it remains to show that $\phi(Puv)\in\langle\mathscr R_v^p\rangle$ holds.
  Note that $Pu$ belongs to $\mathscr P_u^p$, so the induction hypothesis implies
  \( \phi(Pu)\in\langle \mathscr R_u^p\rangle\,.\)
  Thus, there are coefficients $a_1,\ldots,a_d\in F$ and paths $R_1,\ldots, R_d\in\mathscr R_u^p$ such that
  \begin{equation}\label{eq: phi linear comb}
    \phi(Puv)= \phi(Pu)\wedge \phi(v) =  
    \left(\sum_{j=1}^d a_i\phi(R_i)\right) \wedge \phi(v) =
    \sum_{j=1}^d a_i\phi(R_i) \wedge \phi(v) =
    \sum_{j=1}^d a_i \phi(R_iv)\,.
  \end{equation}
  From $R_i\in \mathscr R_u^p$ and $uv\in E(G)$ we obtain $R_iv\in \mathscr Q_v^{p+1}$ by construction~\eqref{eq: Q def}.
  If~$R_iv$ is not a path, then $\phi(R_i v)=0\in\langle \emptyset \rangle$ holds, which implies that $R_iv$ is not added to $\mathscr R_v^{p+1}$ in Step~R3.
  Thus Step~R3 only ever adds paths, which implies~$\mathscr R_v^{p+1}\subseteq\mathscr P_v^{p+1}$ as required.
  Even if $R_iv$ is path, we may not have~$R_i v\in\mathscr R_v^{p+1}$.
  Nevertheless Step~R3 ensures that $\phi(R_iv)\in \langle \mathscr R_v^{p+1}\rangle$ holds at the end of the construction.
  Together with expression~\eqref{eq: phi linear comb}, this shows that $\phi(Puv)$ belongs to $\langle \mathscr R_v^{p+1}\rangle$, so that the representation invariant \eqref{eq: inv rep 1} holds.
  
  It remains to prove the correctness of the algorithm.
  If the algorithm outputs true, then $\emptyset\ne\mathscr R_v^k\subseteq\mathscr P_v^k$ holds, and so there exists a $k$-path.
  On the other hand, if there exists some $k$-path, say $P\in \mathscr P_v^k$, then
  $\phi(P)\ne 0$ follows from Lemma~\ref{lem: main} and the fact that the extensors~$\phi(v_i)$ are in general linear position.
  We have $\phi(P)\in \langle \mathscr R_v^k\rangle$ by~\eqref{eq: inv rep 1}, which implies that $\dim\langle\mathscr R_v^k\rangle\neq 0$ and $\mathscr R_v^k\neq\emptyset$ holds.
  Thus the algorithm correctly outputs 'true'.

  For the running time, computation of $\mathscr R_v^p$ requires linear algebra on $2^k\times 2^kn$ matrices over $F$.
  This can be done in time $\exp(O(k))\poly(n)$.
\end{proof}

A more careful analysis of the linear and exterior algebra operations yields  an upper bound of $2^{\omega k}\poly 
(n)\leq 5.19^k \poly(n)$ on the running time of algorithm R.

\subsubsection*{Acknowledgments}
  We thank Markus Bl\"aser, Radu Curticapean, Balagopal Komarath, Ioannis Koutis, Pascal Schweitzer, Karteek Sreenivasaiah and Meirav Zehavi for some valuable discussions and insights.
  Part of this work was done while the authors attended Dagstuhl seminar~17341, \emph{Computational Counting}.
  Thore Husfeldt is supported by the Swedish Research Council grant VR-2016-03855 and the Villum Foundation grant 16582.

\bibliographystyle{plainurl}
\bibliography{extensor-coding}

\appendix

\section{Random Determinants}
\label{sec: random matrix}
Expressions for the higher moments of determinants of random matrices are available in the literature since the 1950s, see \cite{Nyquist54}.
Such results are considered routine, and follow from exercise 5.64 in Stanley \cite{Stanley} or the general method laid out on page 45--46 in Girko's book \cite{Girko}, but we have found no presentation that is quite complete.
For a judicious choice of distribution, the arguments become quite manageable, so we include a complete derivation.

\medskip
Let $B$ denote a random $k\times k$ matrix constructed by choosing every entry independently and at random from the set $\{\pm \surd 3,0\}$ with the following probabilities:
\[\Pr(b_{ij} = -\surd 3) = \Pr(b_{ij}= \surd 3) = \textstyle\frac{1}{6},\qquad 
\Pr(b_{ij}= 0) = \frac{2}{3}\,.\]

It is clear that every matrix entry satisfies $\E b_{ij} = 0$ and $\E b_{ij}^2=\frac{1}{3} 3 = 1$ and $\E b_{ij}^4= \frac{1}{3} 9 = 3$.

We will investigate the second and fourth moments of $\det B$.
By multiplicativity of the determinant, we can write $(\det B)^r = \det B^r$.

To see
\begin{equation}
  \label{eq: 2nd moment concrete}
  \E \det B^2= k!
\end{equation}
expand $\det B$ by the first row. 
If we write $B_{ij}$ for $B$ with the $i$th row and $j$th column deleted, we have
\[
\E \det B^2 = \E \sum_{i,j}(-1)^{i+j} b_{1i}b_{1j} \det B_{1i}\det B_{1j}\,.\]
The sum extends over all choices of $i,j\in\{1,\ldots,k\}$, but the only nonzero contributions are from $i=j$.
This is because for $i\neq j$, the factor $b_{1j} \det B_{1i}\det B_{1j}$ depends only on variables that are independent of $b_{1i}$, and the latter vanishes in expectation.
Thus,
\[ \E\det B^2= \sum_i (-1)^{2i} \E b_{1i}^2 \E\det B_{1i}^2 = k\E\det B_{11}^2\,,\]
because the distributions of $\det B_{1i}$  for $i\in\{1,\ldots,k\}$ are the same.
This can be viewed as a recurrence relation for $\E\det B^2$ as a function of the dimension $k$, which solves to \eqref{eq: 2nd moment concrete}.

\medskip

To show \[ \E \det B^4 = \tfrac{1}{2} (k!)(k+1)(k+2),\,\]
we use the same kind of arguments.
Write $f_k$ for $\E\det B^4$.
We have \( f_1 = \E b_{11}^4 = 3\) and can compute \( f_2 = 12 \). 
For larger $k$, we expand the first row of $B$ to obtain
\begin{equation*}
  f_k=
  \E \det B^4 = 
  \E \sum_{i,j,l,m}(-1)^{i+j+l+m} b_{1i}b_{1j}b_{1l}b_{1m} \det (B_{1i} B_{1j} B_{1l} B_{1m})\,.
\end{equation*}
As before, if any of $\mset{i,j,l,m}$ differs from the others, the corresponding term vanishes. 
The surviving contributions are of two kinds.
Either $i=j=l=m$, in which case the contribution is 
\begin{equation}\label{eq: contribution 1}
  \sum_i (-1)^{4i} \E b_{1i}^4 \E\det B_{1i}^4 = 3k\E\det B_{11}^4 = 3kf_{k-1}\,.
\end{equation}
Otherwise there are 3 ways in which the multiset $\mset{i,j,l,m}$ consists of two different pairs of equal indices.
The total contribution from these cases is
\begin{equation}\label{eq: contribution 2}
3\sum_{i\neq j} (-1)^{2i + 2j} \E b_{1i}^2 \E b_{1j}^2 \E\det ( B_{1i}^2B_{1j}^2 )
  = 3k(k-1)\E \det ( B_{11}^2B_{12}^2 )\,.
\end{equation}
We continue by expanding $B_{11}$ and $B_{12}$ along their first column.
This is the second and first column, respectively, of the original~$B$.
To keep the index gymnastics manageable, we briefly need the notation $B_{I,J}$ for $B$ without the rows in $I$ and the columns in~$J$.

The nonzero contributions are
\begin{equation*}
  \E \det (B_{11}^2 B_{12}^2) = 
  \E \sum_{i=2}^k\sum_{j=2}^k b_{i2}^2 b_{j1}^2 \det B_{\{1,i\},\{1,2\}}^2 \det B_{\{1,j\},\{2,1\}}^2\,.
\end{equation*}
We note that both $b_{i2}^2$ and $b_{j1}^2$ appear independently, because the remaining submatrices avoid the first and second columns of $B$. 
Since their expectations are unity, they can be removed from the expression.
For $i=j$, both matrices are the same, and the expression collapses to $(k-1)f_{k-2}$.
For $i\neq j$, we introduce the shorthand
\[ \Phi = \E\det (B_{\{1,i\},\{1,2\}}^2B_{\{1,j\},\{2,1\}}^2)\, \qquad(i\neq j)\,,	\]
observing that all these distributions are the same.
We arrive at
\begin{equation}\label{eq: E of smaller determinants}
  \E \det (B_{11}^2 B_{12}^2) = (k-1)f_{k-2} + (k-1)(k-2)\Phi\,.
\end{equation}
Combining \eqref{eq: contribution 1}, \eqref{eq: contribution 2}, and \eqref{eq: E of smaller determinants}, we obtain 
\begin{equation}\label{eq: f_k}
  f_k = 3kf_{k-1} + 3k(k-1)^2 \bigl( f_{k-2} + (k-2)\Phi\bigr)\,.
\end{equation}
Using similar arguments from a different starting point, we obtain
\begin{equation}\label{eq: f_k-1}
f_{k-1} = \E \det B_{11}^4 = 3(k-1)f_{k-2} + 3(k-1)(k-2)\Phi\,,
\end{equation}
by expanding $B_{11}$ along the second column;
the manipulations rely on the fact that in the definition of $\Phi$, the order in which columns $1$ and $2$ are deleted plays (of course) no role.
Combining \eqref{eq: f_k} and \eqref{eq: f_k-1} yields
\begin{equation*}
  f_k = k(k+2)f_{k-1}\,,
\end{equation*}
which solves to $f_k = \frac{1}{2}(k!)^2(k+1)(k+2)$.

\begin{rem}
  We can use this distribution in Section~\ref{sec: Bernoulli coding} in place of the uniform distribution on $\{+1,-1\}$.
The only thing we have to keep in mind is that we still have to be able to perform arithmetic operations in the field.
While this is clear for $\pm 1$, our use of irrational numbers here might create some confusion.
However, note that we only have to calculate with values coming from the field extension $\QQ[\sqrt 3]$, which can be handled just like complex numbers in spirit (after all, $\CC = \RR[\sqrt{-1}]$)
by representing a number $a + b \sqrt 3$ by the two rational coordinates $a,b$,
and performing multiplication according to $(a+b\sqrt 3) (c + d\sqrt 3) = ac + 3bd + (ad+bc)\sqrt 3$.
\end{rem}

\section{Generalization to Subgraphs}

In this section, we formally prove Theorem~\ref{thm approx count sub}.
We use the homomorphism polynomial as a tool for the computation, and we will evaluate this polynomial over a commutative algebra~$\algebra$ analogous to how this was done in Section~\ref{sec: Bernoulli coding}.
For two graph~$H$ and~$G$, let $\Hom HG$ be the set of all functions~$h:V(H)\to V(G)$ that are graph homomorphisms from~$H$ to~$G$.
Then the following is the homomorphism polynomial of~$H$ in~$G$:
  \begin{equation}\label{eq: hom polynomial}
    \sum_{h\in\Hom HG} \prod_{v\in V(H)} \zeta_{h(v)}
    \,.
  \end{equation}
  The variables are $\zeta_v$ for all $v\in V(G)$.
  We first show in \S\ref{app:hompoly} that this polynomial has small algebraic
  circuits when the pathwidth or the treewidth is bounded, and in \S\ref{app:proof of approx count sub} we prove
  Theorem~\ref{thm approx count sub}.

\subsection{Tree Decompositions}
\label{app:hompoly}

Fomin \emph{et al.}~\cite[Lemma~1]{DBLP:journals/jcss/FominLRSR12} construct an algebraic circuit that computes the homomorphism polynomial, based on a dynamic programming algorithm (\emph{e.g.},~\cite{DBLP:journals/tcs/DiazST02}).
We reproduce a proof for completeness.
\begin{lem}\label{lem: hom polynomial}
  Let~$H$ and~$G$ be graphs with $V(H)=\set{1,\dots,k}$ and $V(G)=\set{1,\dots,n}$.
  There is an algebraic circuit~$C$ of size $O\paren[\big]{k\cdot n^{\tw H+1}}$ (and an algebraic skew-circuit~$C$ of size $O\paren[\big]{k\cdot n^{\pw H+1}}$) in the variables~$\zeta_1,\dots,\zeta_n$ such that~$C$ computes the homomorphism polynomial of~$H$ in~$G$ in the variables~$\zeta_1,\dots,\zeta_n$, that is, we have
  \begin{equation}\label{eq: circuit hom polynomial}
    C(\zeta_1,\dots,\zeta_n)
    =
    \sum_{h\in\Hom HG} \prod_{v\in V(H)} \zeta_{h(v)}
    \,.
  \end{equation}
  The circuit can be constructed in time~$O(1.76^k) + |C| \cdot \operatorname{polylog}(|C|)$.
\end{lem}

We first need some preliminaries on tree decompositions.
A \emph{tree decomposition} of a graph~$G$ is a pair $(T,\bag)$, where $T$ is a tree and $\bag$ is a mapping from~$V(T)$ to $2^{V(G)}$ such that, for all vertices $v\in V(G)$, the
set $\setc{t\in V(T)}{v\in\bag(t)}$ is nonempty and connected in~$T$, and for all edges $e\in E(G)$, there is some node $t\in
V(T)$ such that $e\subseteq\bag(t)$.
The set $\bag(t)$ is the \emph{bag at~$t$}.
The \emph{width} of $(T,\bag)$ is the integer $\max\setc{|\bag(t)|-1}{t\in
V(T)}$, and the {\em treewidth} $\tw{G}$ of $G$ is the minimum possible
width of any tree decomposition of~$G$.

It will be convenient for us to view the tree $T$ as being
directed away from the root, and we define the following mappings
$\sep,\cone,\comp:V(T)\to2^{V(G)}$ for all $t\in V(T)$:
\begin{align}
  \label{eq:bagsep}
\text{(\emph{separator at $t$}) \quad }
  \sep(t)&=
  \begin{cases}
    \emptyset\,,&\text{$t$ is the root of $T$},\\
    \bag(t)\cap\bag(s)\,,&\text{$s$ is the parent of $t$ in $T$},
  \end{cases}\\
  \text{(\emph{cone at $t$}) \quad }
  \label{eq:bagcone}
  \cone(t)&=\bigcup_{\text{$u$ is a descendant of $t$}}\bag(u),\\
  \text{(\emph{component at $t$}) \quad }
  \label{eq:bagcomp}
  \comp(t)&=\cone(t)\setminus\sep(t).
\end{align}

\begin{proof}[Proof of Lemma~\ref{lem: hom polynomial}]
  We first compute a minimum-width tree decomposition~$(T,\bag)$ of~$H$, for example using the $O(1.76^k)$ time algorithm by Fomin and Villanger~\cite{DBLP:journals/combinatorica/FominV12}.
    We can assume it to be a \emph{nice} tree 
  decomposition, in which each node has at most two children; the leaves satisfy 
  $\bag(v)=\emptyset$, the nodes with two children $w,w'$ satisfy
  $\bag(v)=\bag(w)=\bag(w')$ and are
  called \emph{join} nodes, the nodes with one child~$w$ satisfy 
  $\bag(v)=\bag(w)\cup\{x\}$ and are called \emph{introduce} nodes, or
  $\bag(v)\cup\{x\}=\bag(w)$ and are
  called \emph{forget} nodes.

  Recall that $\cone(v)$ is the union of all bags at 
  or below node~$v$ in the tree~$T$.
  Let $S\subseteq V(H)$ and $\pi\in\Hom{H[S]}{G}$ be a partial homomorphism from~$H$ to~$G$.
  To make the inductive definition of the algebraic circuit easier, we define the \emph{conditional homomorphism polynomial} as follows.
  \begin{equation}
    \hom\parenc{H\to G}{\pi}
    =
    \sum_{\substack{h\in\Hom HG\\h\supseteq\pi}} \prod_{v\in V(H)} \zeta_{h(v)}\,.
  \end{equation}
  The sum is over all homomorphisms~$h$ that extend~$\pi$.
  With this definition, it is clear that
  \begin{equation}\label{eq:homomorphism conditioning}
    \hom\paren{H\to G}
    =
    \sum_{\pi\in\Hom{H[S]}{G}}
    \hom\parenc{H\to G}{\pi}
  \end{equation}
  holds.
  Moreover, if~$S$ is a separator of~$H$, the connected components of $H-S$ conditioned on the boundary constraints~$\pi$ are independent.
  More precisely, for all~$\pi\in\Hom{H[S]}{G}$, we have
  \begin{equation}\label{eq:homomorphism independence}
    \hom\parenc{H\to G}{\pi}
    =
    \prod_{i} \hom\parenc{H_i\to G}{\pi}\,,
  \end{equation}
  where $H_1,\dots,H_\ell$ is a list of graphs such that $H_1\cup\dots\cup H_\ell=H$ holds, $V(H_i)\cap V(H_j)=S$ holds for all $i,j$ with $i\ne j$, and $H_i-S$ is connected for all~$i$.

  We will construct the final circuit recursively over the tree decomposition~$(T,\beta)$.
  At node~$v$ of~$T$, we construct algebraic circuits~$C_v^\pi$ for each~$\pi\in\Hom{\bag(v)}{G}$ such that the following holds:
  \begin{equation}\label{eq:circuit inductive assumption}
    C_v^\pi
    =
    \hom\parenc{H[\cone(v)]\to G}{\pi}\,.
  \end{equation}
  Note already here that there are at most $n^{\abs{\bag(v)}}\le n^{\tw H+1}$ such functions~$\pi$.
  Since each $C_v^\pi$ represents a gate in our final circuit, and we will be able to charge at most a constant number of wires to each gate, the number of gates and wires of~$C$ will be bounded by $O(\abs{V(H)}\cdot n^{\tw H + 1})$.

  \paragraph{Leaf nodes.}
  Let $v$ be a leaf of~$T$, which has ${\cone(v)=\emptyset}$, resulting in the trivial circuit~$C_v^\pi = 1$ for the empty function~${\pi: \emptyset\to V(G)}$.
  Indeed, since $H[\cone(v)]$ has no vertices, the empty function is the unique homomorphism into~$G$.

  \paragraph{Introduce nodes.}
  Let $v$ be an \emph{introduce} node of~$T$.
  Let $w$ be its unique child in the tree.
  Suppose the vertex~$x\in V(H)$ is introduced at this node, that is, 
  $\bag(w)\cup\{x\}=\bag(v)$.
  Let $\pi\in\Hom{H[\bag(v)]}{G}$ be a partial homomorphism at the bag of~$v$.
  Then we define $C_v^\pi$ using the circuit~$C_w^{\pi'}$ where $\pi'=\pi\restriction_{\bag(w)}$:
  \begin{equation}\label{eq:circuit at introduce nodes}
    C_v^\pi = C_w^{\pi'}\cdot \zeta_{\pi(x)}\,.
  \end{equation}
  For the correctness, note that the right side of~\eqref{eq:circuit at forget nodes} is equal to 
  \begin{equation}\label{eq:introduce correctness}
    \hom\parenc{{H[\cone(w)]}\to{G}}{\pi'} \cdot \zeta_{\pi(x)}
  \end{equation}
  by the induction hypothesis~\eqref{eq:circuit inductive assumption}.
  Since $x$ is the unique vertex in $\cone(v)\setminus\cone(w)$ and $\pi$ extends $\pi'$ on~$x$, the polynomial in~\eqref{eq:introduce correctness} is equal to
  $\hom\parenc{{H[\cone(v)]}\to{G}}{\pi}$.

  \paragraph{Forget nodes.}
  Let $v$ be a \emph{forget} node of~$T$.
  Let $w$ be its unique child in the tree.
  Suppose the vertex~$x\in V(H)$ is forgotten at this node, that is, 
  $\bag(w)\setminus\{x\}=\bag(v)$.
  Then the neighborhood of $x$ is contained in $\cone(w)$.
  Let $\pi\in\Hom{H[\bag(v)]}{G}$.
  We define $C_v^\pi$ using the circuits~$C_w^{\pi'}$ as follows:
  \begin{equation}\label{eq:circuit at forget nodes}
    C_v^\pi = \sum_{\pi'} C_w^{\pi'}\,.
  \end{equation}
  The sum is over all~$\pi'\in\Hom{H[\bag(w)]}{H}$ that agree with~$\pi$ on the intersection~$\bag(v)\cap\bag(w)$ of their domains.
  Since~$v$ is a forget node, this intersection is equal to $\bag(v)$.
  For the correctness, note that the right side of~\eqref{eq:circuit at introduce nodes} is equal to
  \begin{equation}
    \sum_{\pi'} \hom\parenc{{H[\cone(w)]}\to{G}}{\pi'}
  \end{equation}
  by the induction hypothesis~\eqref{eq:circuit inductive assumption}.
  This is equal to $\hom\parenc{{H[\cone(v)]}\to{G}}{\pi}$ due to the conditioning formula~\eqref{eq:homomorphism conditioning}.
  For the size bound, note that~$x$ is the only vertex in $\bag(w)\setminus\bag(v)$.
  Thus, the sum in~\eqref{eq:circuit at forget nodes} has~$n$ terms, and so in this part of the construction we charge at most one wire to each $C_w^{\pi'}$.

  \paragraph{Join nodes.}
  Let $v$ be a \emph{join} node of~$T$ with 
  children~$w$ and~$w'$ satisfying $\bag(v)=\bag(w)=\bag(w')$.
  Let ${\pi\in\Hom{H[\bag(v)]}G}$.
  We define the circuit simply as
  \begin{equation}
    C_v^\pi = C_w^\pi \cdot C_{w'}^\pi\,.
  \end{equation}
  For the correctness, note that $\bag(v)$ is a separator for $H[\cone(v)]$, and so the induction hypothesis~\eqref{eq:circuit inductive assumption} together with the conditional independence~\eqref{eq:homomorphism independence} yields the correctness.
  This part of the construction introduces two wires which we charge to $C_v^\pi$.

  The final circuit is $C=C_r^\emptyset$, where $r$ is the root of~$T$, the degenerate empty function is~$\emptyset$, and we assume without loss of generality that $\bag(r)=\emptyset$.
  As already discussed, we have $O(\abs{V(T)} n^{\tw H + 1}) = O(\abs{V(H)} n^{\tw H+1})$ gates $C_v^\pi$, each of which is responsible for $O(1)$ wires incident to it.

  Finally, note that if there are no join nodes, then $(T,\beta)$ is a path decomposition of~$H$ and the only multiplications occur in~\eqref{eq:circuit at introduce nodes} and involved at least one variable.
  Thus, when $(T,\beta)$ is a minimum-width path-decomposition, the algebraic circuit~$C$ constructed above is skew, and has size~$O(k n^{\pw H+1})$.
\end{proof}

\subsection{Proof of Theorem~\ref{thm approx count sub} and~\ref{thm decide sub}}
\label{app:proof of approx count sub}
\begin{repthm}{thm approx count sub}
  \statethmapproxcount
\end{repthm}
\begin{proof}[Proof sketch]
  Let $H$, $G$, and $\varepsilon>0$ be given as input.
  Let $n$ be the number of vertices of~$G$.
  By Lemma~\ref{lem: hom polynomial}, we can construct an algebraic skew circuit~$C$ that computes the homomorphism polynomial of~$H$ in~$G$.
  The circuit has size~$O(k n^{\pw H+1})$ and satisfies:
  \begin{equation}
    C(\zeta_1,\dots,\zeta_n)
    =
    \sum_{h\in\Hom HG} \prod_{v\in V(H)} \zeta_{h(v)}
    \,.
  \end{equation}
  
  Following the setup of \S\ref{sec: Bernoulli coding}, we use the lifted Bernoulli extensor-coding $\overline\beta\colon V(G)\to F^{2k}$.
  We have
  \begin{equation}\label{eq: xi in hompoly}
    C(\overline\beta(v_1),\dots,\overline\beta(v_n))
    =
    \pm \sum_{h\in\InjHom HG} \prod_{v\in V(H)} \overline\beta({h(v)})
    \,,
  \end{equation}
  where $\InjHom HG$ is the subset of~$\Hom HG$ that consists of all homomorphisms that are injective.
  Now we use Algorithm~C, except that we replace $f(G;\overline\beta)$ with
  $C(\overline\beta(v_1),\dots,\overline\beta(v_n))$.
  The rest goes through as in Theorem~\ref{thm: algorithm C}.
  Note that this approach using~\eqref{eq: xi in hompoly} actually approximates the number of injective homomorphisms, which however gives rise to an approximation (of the same quality) for $\Sub(H,G)$ when we divide by the size~$\abs{\operatorname{Aut}(H)}$ of the automorphism group of~$H$.
  The size of the automorphism group of~$H$ can be computed in advance, in time $O(1.01^k)$, by a well-known $\poly(k)$-time reduction to the graph isomorphism problem~\cite{DBLP:journals/ipl/Mathon79}, which in turn can be computed in time $\exp(\poly\log k)\le O(1.001^k)$~\cite{DBLP:conf/stoc/Babai16}.
  For the running time of the modified Algorithm~C, note that~$C$ is a skew circuit, and skew multiplication in~$\Lambda(F^{2k})$ takes time $O(4^k)$.
  Thus, the overall running time is~$O(\varepsilon^{-2}4^k \abs{C})$.
\end{proof}

Theorem \ref{thm decide sub} in the case of paths is established through Theorem \ref{thm: algorithm F}.
For the more complicated case of general subgraphs, we can modify the polynomial \eqref{eq: hom polynomial} analogously so that it involves also the edge variables. By suitable modifications of the dynamic program that computes the corresponding circuit, Theorem \ref{thm decide sub} is established.

\section{Proof of Theorem 3} \label{sec: mult det}
In this section, we prove Theorem~\ref{thm detect multilinear}.
This will follow by an application of our algebraic interpretation of color-coding from Section~\ref{sec: color-coding}. 
In particular, we will again make use of the Zeon algebra $Z(F^k)$.

The next proposition is a trivial consequence of Lemma \ref{lem: zeons}.
\begin{prop} \label{prop: zeon eval}
  For any integer $t>0$, an arithmetic circuit $C$ over $\ZZ[\zeta_1,\ldots,\zeta_n]$ can be evaluated over $Z(\QQ^t)$ in $2^t \cdot |C|\cdot \poly(n)$ operations over $\QQ$.
\end{prop}

We are ready for the proof:

\begin{proof}[Proof of Theorem \ref{thm detect multilinear}]
  We invoke Proposition \ref{prop: zeon eval} with Hüffner \emph{et al.}'s \cite{DBLP:journals/algorithmica/HuffnerWZ08} choice of $t=1.3k$.
One evaluation costs  $2.4623^k \cdot |C|\cdot \poly(n)$ operations over $\QQ$.
The classical color-coding approach would evaluate $C$ at the generators $\overline{\cano i}$,
where $1 \leq i \leq t$ is chosen uniformly at random.
In this way, all non-multilinear terms will vanish, 
but distinct monomials might cancel when being mapped to the same product of $\overline{\cano i}$.
To avoid this, we randomly scale each generator,
and plug in $\alpha_i \cdot \overline{\cano j}$ at the $i$th input of the circuit, for random $\alpha_i \in \{0,1,\ldots,100\cdot k\}$ and random $1 \leq j \leq t$.
The circuit $C$ then evaluates to some multiple of $\overline{\cano t}$,
and the coefficient of $\overline{\cano t}$ in the result is a multilinear polynomial in the $\alpha_i, 1 \leq i \leq n$.
By the DeMillo--Lipton--Schwartz--Zippel-Lemma \cite{DBLP:journals/ipl/DemilloL78, DBLP:journals/jacm/Schwartz80, DBLP:conf/eurosam/Zippel79} the polynomial evaluates non-zero with constant probability of $99\%$.
Following Hüffner \emph{et al.} \cite[Theorem 1]{DBLP:journals/algorithmica/HuffnerWZ08},
the probability that some multilinear term maps to a multiple of 
the $t$ generators is at least $\Omega(1.752^{-k})$, 
and we derive the total running time of $4.32^k\cdot |C|\cdot \poly(n)$ operations in $\QQ$.
Now, if the circuit can be evaluated over $\ZZ$ in polynomial time
(\emph{i.e.}, all numbers stay of appropriate size), this costs only a polynomial overhead, and the claim follows.
However, by repeated squaring, the circuit $C$ may generate numbers of value $2^{2^n}$,
so we calculate modulo some random prime.
Numbers of bitlength~$O({2^n})$ may have up to~$O(2^n)$ prime factors,
so choosing a random prime $p$ from the first $\Omega(n 2^n)$ primes,
we find a value for $p$ such that the resulting coefficient doesn't vanish modulo $p$
with probability $1 - o(1)$.
By the prime number theorem, the first $n 2^n$ primes are of magnitude $2^n \poly(n)$,
and we can thus randomly pick a number from $\{1,\ldots,P\}$,
where $P = 2^n \poly(n)$, until we find a prime (which can be tested in randomized polynomial time).
Then we perform all the above calculations modulo $p$,
and if the result doesn't vanish, the polynomial doesn't vanish over $\ZZ$.
On the other hand, if it vanishes, we might have had bad luck,
but as argued, this only happens with probability $1/n$, which is fine.
\end{proof}
\end{document}